\documentclass[a4paper,12pt]{article}

\usepackage[english]{babel}
\usepackage[T1]{fontenc}
\usepackage[utf8]{inputenc}
\usepackage{amsthm}
\usepackage{amsmath}
\usepackage{graphicx}
\usepackage{amssymb}
\usepackage{latexsym}
\usepackage{color}
\usepackage{url}
\usepackage{subcaption}

\theoremstyle{definition}
\newtheorem{definition}{Definition}
\theoremstyle{plain}
\newtheorem{theorem}[definition]{Theorem}
\newtheorem{proposition}[definition]{Proposition}

\newtheorem{lemma}[definition]{Lemma}

\theoremstyle{definition}

\newcommand{\K}{\mathbb{K}}
\newcommand{\NN}{\operatorname{NN}}
\newcommand{\NT}{\operatorname{NT}}
\newcommand{\TT}{\operatorname{TT}}
\newcommand{\TTT}{\operatorname{TTT}}
\newcommand{\NNNT}{\operatorname{3NT}}
\newcommand{\NNN}{\operatorname{NNN}}

\newcommand{\Neg}{\operatorname{Neg}}

\newcommand{\col}{\operatorname{Col}}

\newcommand{\klein}{\mathbb{Z}_2 \times \mathbb{Z}_2}
\newcommand{\pentagon}{\mathbf{P}}
\newcommand{\Dp}{\mathbf{dP}}
\newcommand{\Tp}{\mathbf{tP}}
\newcommand{\Tc}{\mathbf{tC}}
\newcommand{\qT}{\mathbf{qT}}
\newcommand{\dyad}{\mathbf{D}}
\newcommand{\triad}{\mathbf{T}}
\def\dotp{\,.\,}

\let\phi\varphi
\def\phisum{\phi_*}

\begin{document}

\title{Morphology of small snarks}

\author{
	Ján Mazák, Jozef Rajník, Martin Škoviera
	\\[3mm]
	\\{\tt \{mazak, rajnik, skoviera\}@dcs.fmph.uniba.sk}
	\\[5mm]
	Comenius University, Mlynská dolina, 842 48 Bratislava\\
}

\maketitle

\begin{abstract}
The aim of this paper is to classify all snarks up to order
$36$ and explain the reasons of their uncolourability. The
crucial part of our approach is a computer-assisted structural analysis of
cyclically $5$-connected critical snarks,
which is justified by
the fact that every other snark can be constructed from them by
a series of simple operations
while preserving uncolourability.
Our results reveal that most
of the analysed snarks are built up from pieces of the Petersen
graph and can be naturally distributed into a small number of
classes having the same reason for uncolourability. This sheds
new light on the structure of all small snarks. Based on our
analysis, we generalise certain individual snarks to infinite
families and identify a rich family of cyclically $5$-connected
critical snarks.
\end{abstract}

\bigskip

\section{Introduction}

In this paper we attempt to provide insight into the structure
of snarks of small order. The ultimate aim of our endeavour is,
of course, to contribute to understanding the nature of
snarks in general. Snarks are --- in essence --- connected cubic
graphs whose edges cannot be properly coloured with three
colours. In recent years, snarks have been attracting
considerable attention, mainly  because this family might
contain counterexamples to several profound and long-standing
conjectures such as the cycle double cover conjecture, the
5-flow conjecture, or the Berge-Fulkerson conjecture
\cite{Fulk, J85,J88}. Understanding the structure of snarks is
therefore crucial for proving or disproving any of these
conjectures.

Snarks are very difficult to find since almost all cubic graphs
are hamiltonian and hence $3$-edge-colourable
\cite{rob-wormald}. This asymptotic behaviour manifests itself
already at very small orders: on $26$ vertices, there exist
fewer than three cyclically $4$-edge-connected snarks per
million of cubic graphs, and this ratio seems to exponentially
decrease with increasing order. On the other hand, deciding
whether a cubic graph is $3$-edge-colourable or not is an
NP-complete problem \cite{Holyer}, which means that the class
of snarks is still sufficiently rich: for instance, there are
more than $400$ millions of cyclically $4$-edge-connected
snarks on at most $36$ vertices \cite{BGHM, GMS1}.

Snarks are also difficult to understand because the reasons
that force the absence of $3$-edge-colourings in cubic graphs
are in general unknown. Numerous constructions of snarks have
been presented by various authors (see \cite{AKLM, Isaacs,
Snarks, Oddness, Macajova, Superposition,
SteffenClassifications} for some examples), often aiming at
proving the existence of snarks possessing certain special
properties. Despite this effort, very little is known about the
intrinsic structure of snarks. For example, it is not known
whether small edge-cuts in snarks are unavoidable: a conjecture
of Jaeger and Swart \cite{Jaeger} states that every snark
contains a cycle-separating edge-cut comprising at most six
edges, but this conjecture has been open for almost 40 years
without any visible progress. Thus no general approach to
classification of the entire family of snarks seems to be
within the reach.

The evidence drawn from the published lists of small snarks
(see, for example, \cite{BGHM, BS, CMRS}) reveals that most of
them are composed of several common construction components,
typically obtained from the Petersen graph. This phenomenon has
neither been formalised nor thoroughly studied yet, and it is
unclear whether anything similar holds for snarks of large
order. Although we believe that these questions are well worth
of investigation, it does not seem that the currently available
methods are powerful enough to attack them in full generality.
Neither probabilistic nor constructive methods provide us with
good insight into the structure of snarks. In particular, no
uniform random model for snarks is currently known, which
leaves us without strong theoretical tools for studying the
typical behaviour of snarks. The NP-completeness of the problem
of edge-colourability \cite{Holyer} also indicates that no
quick progress is likely. This is why we focus on what we have
at hand, which is the complete list of all cyclically
$4$-edge-connected snarks on up to $36$ vertices, recently
produced by Brinkmann et al. \cite{BGHM} and completed in
\cite{GMS1}, by employing exhaustive computer search.

Suppose, for a moment, that we would like to move a step
further and produce a list of all snarks on $38$ vertices. An
obvious way to do it would be to generate all cubic graphs on
$38$ vertices and discard those that are colourable.
Unfortunately, there are just too many of them in comparison
with the computing power presently available for research, even
if we restrict ourselves to those with cyclic connectivity at
least $4$ and girth at least~$5$. Such an approach is therefore
very unlikely to work. In this situation it may be useful to
realise that a vast majority of known snarks contain an edge
whose removal followed by the suppression of the resulting
$2$-valent vertices again leaves a snark. Conversely, most
known snarks arise from a smaller snark by choosing a suitable
pair of edges, subdividing each of them with one additional
vertex, and connecting the resulting $2$-valent vertices with a
new edge; this operation is called an I-extension. The meaning
of ``suitable'' in order for the operation of I-extension to be
feasible is easy to explain (see
Proposition~\ref{prop:edge-extension}), which suggests that
this approach might be promising. The hard part of the problem,
however, are the snarks that cannot be obtained by a series of
I-extensions from a smaller snark in such a way that each
member of the extension series is a snark. Such snarks indeed
exist and have been already studied \cite{Chladny,
Chladny-Skoviera-Factorisations, Nedela,
SteffenClassifications}, in fact, they have been rediscovered
several times \cite{dSPL, dVNR, S-full}: they are known as
\textit{critical snarks} and are characterised by the property
that for each edge the inverse of I-extension produces a
colourable graph.

Critical snarks are known to be cyclically $4$-edge-connected
with girth at least $5$ \cite[Proposition 4.8]{Nedela}, and
thus can be regarded as ``proper snarks'' by the usual
standards. A decomposition theory developed by Chladný and
Škoviera in \cite{Chladny-Skoviera-Factorisations} suggests
that critical snarks that possess a cycle-separating
$4$-edge-cut can be explained via the reversal of the
well-known operation of dot product. This is especially true
for bicritical snarks, an important subclass of critical
snarks: every bicritical snark containing a cycle-separating
$4$-edge-cut admits a decomposition into a unique collection of
cyclically $5$-edge-connected bicritical snarks, and
conversely, it can be reconstructed from them by a repeated
application of dot product
\cite[Theorem~C]{Chladny-Skoviera-Factorisations}.

By contrast, a decomposition process along $5$-edge-cuts is
much more complicated \cite{Cameron, Nedela, Preissmann}.
Moreover, it only works in one way and cannot be easily used
for constructing snarks. Indeed, as argued in
\cite[p.~273]{Nedela}, the original snark cannot be
reconstructed from decomposition factors by using any
collection of well-defined simple operations. Thus, from this
point of view, the most fundamental and at the same time most
enigmatic snarks are those that are critical and cyclically
$5$-edge-connected, also known under the term \textit{$5$-simple}
\cite{Nedela}. Since all cyclically $4$-edge-connected snarks
can be obtained from them by applying I-extensions and dot
products, it is natural to start the structural analysis of
snarks by investigating \textit{cyclically $5$-edge-connected
critical snarks}.

Our aim in this paper is, therefore, to analyse and classify
\textit{all} $5$-simple snarks of
order not exceeding 36. The list of such snarks is known and
contains exactly $2110$ graphs. We have extracted it from the
complete list of all nontrivial snarks of order up 36 which was
produced by Brinkmann et al. \cite{BGHM} in 2013. The list of
all critical snarks of order up to 36 was previously compiled
by Carneiro et al. \cite{CdSM}, however, those with cyclic
connectivity at least $5$ have not been singled out.

The method which we apply to the analysis of snarks is similar
to what biologists have been doing for centuries in morphology.
By discovering and investigating more and more species they
have been constantly improving and refining the hierarchy of
organisms, making it more complete with every new species
examined. Our aim is to describe the structure of each
$5$-simple snark of order not
exceeding 36 in a manner comprehensible to a human, with
uncolourability readily verifiable by hand, as opposed to having a
proof that relies on an exhaustive enumeration carried out by a
computer. It transpires that uncolourability of small snarks
can be conveniently explained in terms of multipoles (subgraphs
with dangling edges) and their interconnections. Snarks with
similar structure of multipoles, similar interconnections, and
similar reasons for uncolourability are collected into families.
As we analyse larger and larger snarks, the set of multipoles
with known colouring properties grows. Whenever we encounter a
snark whose uncolourability cannot be fully explained by
previously discovered multipoles, we analyse it, extend the
list of known multipoles, and employ it in the further
analysis.

The advantage of this approach is evident from the fact that
each of the families resulting from our analysis can be easily
turned into an infinite class of snarks in a straightforward
manner: it suffices to substitute the basic construction
blocks, usually taken from the Petersen graph, with those
obtained from larger snarks in a similar manner.

Results of our analysis are summarised in
Section~\ref{sec:results}. Amongst the snarks up to $30$
vertices we have not discovered any example that could not be
easily explained in terms of multipoles arising from the
Petersen graph. A new phenomenon arises on $32$ vertices with
class denoted by 32-A and schematically depicted in
Fig.~\ref{fig:32}. The 7-pole $M_{11}$ contained in these
snarks is perhaps the most interesting specimen of all
--- it is the only multipole that we have not been able to
generalise.

On $34$ and $36$ vertices, there are several families of
interest not described before; all of them contain several
disjoint copies of 5-poles each consisting of a pair of
5-cycles sharing two edges. These 5-poles can be obtained from
the Petersen graph by removing a path of length $2$ and are
known as negators.

Perhaps the most aesthetically pleasing family is illustrated
in Fig.~\ref{fig:34-6c}. The smallest snark from this family can
also be regarded as a cleverly arranged tangle of $6$-cycles
complemented by a 6-pole arising from the Petersen graph by
removing a $6$-cycle; it is displayed in Fig.~\ref{fig:34-6b}.

\begin{figure}
    \centering
    \includegraphics{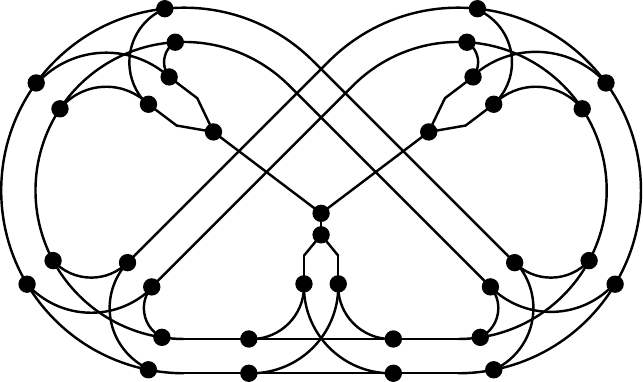}
    \caption{A remarkable critical snark of order 34}
    \label{fig:34-6b}
\end{figure}

Despite our success in explaining the uncolourability of the
investigated graphs, we have not accomplished that much in
terms of truly understanding their criticality. In particular,
we can easily generalise any of the described families of
snarks into an infinite class of new snarks, but we know very
little about which of these new snarks are critical or even
bicritical. All our achievements in this direction are
collected in Section~\ref{sec:newfamily}. On the one hand, we
observe that construction components for critical snarks need
not come from critical snarks. On the other hand, we give
examples of constructions for which obvious necessary
conditions, such as criticality of the constituting multipoles,
are not sufficient to ensure criticality of the resulting
snarks. In spite of that, we are able to describe a new rich
infinite family of bicritical snarks. This family demonstrates
that many (perhaps even all) of the families described in this
paper can be turned into infinite families of bicritical snarks
by imposing additional restrictions on the construction blocks.

Finally, we employ the results of our analysis to give a
negative answer to a question posed by Chladný and Škoviera in
\cite[Problem~5.7]{Chladny-Skoviera-Factorisations} about pairs
of edges essential for a dot product of bicritical snarks to be
bicritical. According to the theory developed therein, an
essential pair of edges must be non-removable (that is, its
removal leaves a colourable graph).  However, it was left open
whether there exists a non-removable pair of edges that is not
essential. An example of such a pair is provided in
Section~\ref{sec:s36}.

We conclude this section with a short list of definitions. We
assume that the reader has the basic knowledge related to graph
colourings and flows. For more information on this matter we
recommend  consulting \cite{J88}.

Our graphs are finite and may contain
parallel edges and loops. A connected $2$-regular graph is
called a \textit{cycle}. A cubic graph $G$ is said to be
\textit{cyclically $k$-connected} (or \textit{cyclically
$k$-edge-connected}, to be more precise) if no set of fewer
than $k$ edges separates two cycles of $G$ from each other.
The \textit{cyclic connectivity} of $G$ is the smallest integer
$k$ such that $G$ is cyclically $k$-connected.
A connected uncolourable cubic graph is
a \textit{snark}. The graph which consists of two vertices
joined by an edge and has a loop attached at each vertex is
called the \textit{dumbbell graph}; it is denoted by~$Db$. Note
that $Db$ is a snark according to our definition. It is well
known that the smallest $2$-connected snark is the Petersen
graph, denoted here by $Pg$. Cyclically $4$-connected snarks
with girth at least $5$ will be called \emph{nontrivial} and
the remaining ones will be \emph{trivial}.

\section{Multipoles and their Tait colourings}

Snarks are often described as combinations of graph-like
structures called multipoles. In contrast to graphs, multipoles
are permitted to contain dangling edges or even isolated edges,
see e. g. \cite{Fiol, Macajova, Nedela}. Formally, a
\emph{multipole} is a pair $M = (V(M), E(M))$, where $V(M)$ is
a set of vertices and $E(M)$ is a set of edges. Every edge $e
\in E(M)$ has two \emph{ends} which may, or may not, be
incident with a vertex. An edge whose ends are incident with
two distinct vertices is called a \emph{link}. If only one end
of an edge is incident with a vertex, then the edge is a
\emph{dangling edge}, and if neither end of an edge is incident
with a vertex, it is called an \emph{isolated edge}. A
\emph{semiedge} is an end of an edge that is incident with no
vertex. The set of all semiedges of a multipole $M$ is denoted
by $S(M)$. A multipole with $k$~semiedges is called a
\emph{$k$-pole}. The \emph{order} $|M|$ of a multipole $M$ is
the number of its vertices. In this paper, we will only
consider \textit{cubic multipoles}, that is, multipoles where
each vertex is incident with three edge ends.

It is often convenient to partition the set $S(M)$ into
pairwise disjoint sets $S_1, \dots, S_n$ called
\emph{connectors}.
Although semiedges in a connector are unordered, sometimes it is 
useful to endow a connector with a linear order. 
Such a connector $S = (e_1, \dots, e_k)$ is called 
an \emph{ordered connector}.
A multipole $M$ with $n$ connectors $S_1, S_2, \dots, S_n$ such
that $|S_i| = c_i$ for $i \in \{1,2,\dots,n\}$ is denoted by
$M(S_1,S_2,\dots,S_n)$ and called a
\emph{$(c_1,c_2,\dots,c_n)$-pole}. If a connector $S$ contains
only one semiedge $s$, we will usually only write $s$ in place
of $\{s\}$.

The \emph{junction} of semiedges $e$ and $f$ is an operation
under which the end-vertices of the corresponding dangling edges are
joined to produce a new edge. 
The junction of two connectors $S = \{e_1,e_2,\dots,e_k\}$ and $T = \{f_1,f_2,\dots,f_k\}$ of the same size $k$ consists of performing $k$ individual junctions of $e_i$ and
$f_i$ for $i~\in~\{1,2,\dots,k\}$.
If $S$ and $T$ are not ordered, prior to performing the junction we can enumerate their semiedges in an arbitrary order. Although different orderings may lead to several different multipoles, in a vast majority of cases our results do not depend on the order in which the junctions of semiedges of two connectors have been performed. 
A junction of two $(c_1, \dots, c_n)$-poles $M(S_1, \dots, S_n)$ and $N(T_1, \dots, T_n)$ consists of $n$ individual junctions of $S_i$ and $T_i$ for each $i~\in~\{1,2,\dots,n\}$. Hence, performing a junction of $M$ and $N$ requires
to join the corresponding connectors $S_i$ and $T_i$, while the order of semiedges in $S_i$ and $T_i$ may be arbitrary, unless the connectors are ordered.

A natural approach to explaining the uncolourability of a snark
is to split it into a set of multipoles and to study
interactions between their colourings. The aim is to show that
any combination of colourings of the constituting multipoles
gives rise to a confict within the snark. 
By an \textit{edge colouring} of a multipole $M$ we mean a
mapping $\phi\colon E(M)\to X$ from the edge set of $M$ to a
certain set $X$ of colours. An edge colouring naturally induces
a colouring of edge ends. If the ends of all edges incident
with any vertex $v$ of $M$ receive distinct colours, the
colouring is said to be \textit{proper}. If $|X|=k$, the
colouring is a \textit{$k$-edge-colouring}.

Since multipoles in this paper are all cubic, colourings
considered in this paper will mostly  be proper
$3$-edge-colourings. This permits us to abbreviate the term
``proper $3$-edge-colouring'' to just ``colouring''.  We
therefore say that a multipole $M$ is \emph{colourable}
whenever it has a colouring; otherwise $M$ is
\emph{uncolourable}.

A convenient set of colours for the study of snarks is provided
by the set $\K$ of nonzero elements of the Klein four-group
$\klein$. An edge colouring using $\K$ as the colour set is
often termed a \textit{Tait colouring} because the usage of
such colourings can be traced back to Tait's paper \cite{Tait}
on the Four-Colour Problem. One of the advantages of using $\K$
is that we can use addition in the group $\klein$ to express
the properties of a colouring. Indeed, a colouring of a
multipole $M$ is proper if and only if for every vertex $v$ of
$M$ the three colours meeting at $v$ sum to $0$ in $\klein$. In
other words, a proper $3$-edge-colouring of a cubic multipole
is a nowhere-zero $\klein$-flow and vice versa. Now, if we
regard a colouring $\phi$ of a multipole $M$ as a flow, we can
use the Kirchhoff law to conclude that $\sum_{e\in S(M)}
\phi(e) = 0$. This fact has a useful consequence commonly known
as the parity lemma, first proved by Tutte  \cite{Descartes} under
the pseudonym of Blanche Descartes.

\begin{lemma}[Parity lemma]\label{lemma:parity}
	Let $M$ be a $k$-pole and let $k_1$, $k_2$, and $k_3$ be the
	numbers of semiedges of colour $(0,1)$, $(1,0)$, and
	$(1,1)$, respectively. Then
	$$ k_1 \equiv k_2 \equiv k_3 \equiv k \pmod{2}.$$
\end{lemma}

\section{Reducibility and criticality of snarks}
\label{sec:reducibility}

We have already indicated why critical snarks are a natural
place to start investigation of the intrinsic structure of
snarks. In this section we discuss this matter in a greater
detail and explain important relations between criticality of
snarks, their nontriviality, and reducibility.

A natural way of approaching the idea of nontriviality of
snarks is by asking whether the snark in question contains
vertices that do not contribute to its uncolourability.
According to the parity lemma, removing just one vertex from a
snark leaves an uncolourable graph, so one has to remove at
least two vertices to make the graph colourable. A pair of
distinct vertices $\{u,v\}$ of a snark $G$ will be called
\textit{removable} if $G-\{u,v\}$ is not $3$-edge-colourable;
otherwise it will be called \textit{non-removable}. A snark $G$
is \textit{critical} if every pair of adjacent vertices in $S$
is non-removable; it is called \textit{bicritical} if every
pair of distinct vertices in $G$ is non-removable.

The concept of criticality of snarks can also be
developed by interpreting $3$-edge-colourings of a cubic graph
as nowhere-zero flows. This direction has been explored by
several authors, see for example \cite{CdSM, dSL, dSPL, dVNR,
FMS-survey}. Following da Silva et al. \cite{dSL} we define a
graph to be \textit{$k$-flow-edge-critical} if it does not
admit a nowhere-zero $k$-flow but the graph obtained by the
contraction of any edge does. We further define a graph to be
\textit{$k$-flow-vertex-critical} if it does not admit a
nowhere-zero $k$-flow but the graph obtained by the
identification of any two distinct vertices does. These
definitions apply to snarks with $k=4$ since no snark admits a
nowhere-zero $4$-flow. If we take into account the fact that
contracting an edge has the same effect on the existence of a
nowhere-zero flow as identifying its end-vertices,
$4$-flow-edge-critical snarks and $4$-flow-vertex-critical
snarks are natural counterparts of critical and bicritical
snarks, respectively. Nevertheless, it has been only recently
shown \cite{MS:crit} that in spite of different formal
definitions flow-critical snarks are exactly the same as
critical snarks.

\begin{theorem}\label{thm:flow-crit}
A snark is $4$-flow-edge-critical if and only if it is
critical. A snark is $4$-flow-vertex-critical if and only if it
is bicritical.
\end{theorem}

Consider a pair of adjacent vertices $u$ and $v$ forming an
edge $e$ of a cubic graph $G$, and let $G\sim e$ denote the
cubic graph homeomorphic to $G-e$. The operation that
transforms $G$ into $G\sim e$ is called an
\textit{edge reduction}, and  its reverse is called an
\textit{edge extension} or an I-\textit{extension}. To perform
an edge extension of a cubic graph $G'$ one picks in $G'$ two
edges $e_1$ and $e_2$ (not necessarily distinct), subdivides
each of them with a new vertex, and adds an edge $e$ joining
the two new vertices (if $e_1=e_2$, this results in a digon,
that is, a pair of parallel edges). The resulting graph is
denoted by $G'(e_1,e_2)$.

It turns out that removing a pair of adjacent vertices $u$ and
$v$ from a snark has the same effect on colourability as
reducing the edge $e$ joining them. This fact was first
observed in \cite{Nedela}.

\begin{proposition}\label{prop:edge-reduction}
Let $G$ be a snark and $e=uv$ an edge of $G$. Then $G\sim e$ is
$3$-edge-colourable if and only if $G-\{u,v\}$ is
$3$-edge-colourable.
\end{proposition}

Proposition~\ref{prop:edge-reduction} implies that a snark $G$
is critical if and only if \mbox{$G\sim e$} is colourable for
each edge $e$. Every noncritical snark thus contains an edge
whose reduction leaves a smaller snark. If the resulting snark
is still not critical, we can repeat the process and continue
until we eventually produce a critical snark. Reversing this
process shows that every snark can be constructed from a
critical snark by a series of edge extensions, with all
intermediate graphs being snarks that contain a subdivision of
the initial critical snark. All this tells us that critical
snarks can be regarded as basic building blocks of all snarks.

In order to make the extension process work, it is important to
know under what conditions an edge extension $G(e_1,e_2)$ of a
snark $G$ is again a snark. The answer requires one more
definition. A pair $\{e_1,e_2\}$ of edges of a snark $G$ is
said to be \textit{removable} if $G-\{e_1,e_2\}$ is
$3$-edge-colourable.

\begin{proposition}\label{prop:edge-extension}
Let $G$ be a snark, and let $e_1$ and $e_2$ be distinct edges
of $G$. Then the edge extension $G(e_1,e_2)$ of $G$ is a
snark if and only if the pair $\{e_1,e_2\}$ is removable.
\end{proposition}

\begin{proof}
If $\{e_1,e_2\}$ is a removable pair of edges of $G$, then
$G-\{e_1,e_2\}$ is uncolourable. Since $G-\{e_1,e_2\}\subseteq
G(e_1,e_2)$, we conclude that so is $G(e_1,e_2)$. Thus
$G(e_1,e_2)$ is a snark.

Conversely, let $G(e_1,e_2)$ be a snark and suppose to the
contrary that $\{e_1,e_2\}$ is non-removable, that is,
$G-\{e_1,e_2\}$ is colourable. Let $e_1'$, $e_1''$ and $e_2'$,
$e_2''$ be the edges of $G(e_1,e_2)$ obtained by subdividing
$e_1$ and $e_2$, respectively. Every 3-edge-colouring $\phi$ of
$G-\{e_1,e_2\}$ forces at least one of the pairs
$\{e_1',e_1''\}$ and $\{e_2',e_2''\}$ to receive distinct
colours, say $\phi(e_1')=a$ and $\phi(e_1'')=b$, otherwise
$\phi$ would induce a colouring of $G$. By the parity lemma, the
other pair also receives colours $a$ and~$b$. Thus the edge of
$G(e_1,e_2)$ added across $e_1$ and $e_2$ can be coloured $a+b$
to produce a 3-edge-colouring of $G(e_1,e_2)$, which is a
contradiction.
\end{proof}

Another possibility to capture the notion of nontriviality of
snarks is to identify edge-cuts whose removal from a snark
produces an uncolourable component. The aim is to generalise
the well-known fact that snarks with short cycles and small
edge-cuts are just trivial modifications of smaller snarks. For
this purpose Nedela and \v Skoviera \cite{Nedela} proposed the
following definitions. Let $G$ be a snark which can be
expressed as a junction $M*N$ of two $k$-poles $M$ and~$N$ for
some $k\ge 0$. If one of $M$ and~$N$, say $M$, is uncolourable,
we can extend $M$ to a snark $\bar M$ of order not greater than
$|G|$ by adding to $M$ a small number of vertices and edges;
possibly $\bar M=G$. By creating the graph $\bar M$ we have
reduced a snark $G$ to a new snark which is called a
\emph{$k$-reduction} of $G$. (Note that an edge reduction
$G\sim e$ is a special case of a $4$-reduction.) A
$k$-reduction $\bar M$ of $G$ is \textit{proper} if $|\bar
M|<|G|$. If $G$ admits a proper $k$-reduction for some $k\ge
0$, the essence of uncolourability of the smaller snark is the
same as the one that can be found in $G$. A snark is
\textit{$k$-irreducible}, for $k\ge 1$, if it has no proper
$m$-reduction for any $m<k$. A snark is \textit{irreducible} if
it is $k$-irreducible for every $k>0$, that is, if it admits no
proper reductions at all. Observe that a $k$-irreducible snark
is also $r$-irreducible for every $r\leq k$.

The following theorem, proved in \cite{Nedela}, puts the
concept of criticality of snarks into the perspective of
various ranks of irreducibility. Among others, it tells us that
there are, surprisingly, only finitely many different degrees
of irreducibility, with bicritical snarks holding the highest
position.

\begin{theorem}\label{thm:k-irred}
Let $G$ be a snark. Then the following statements hold true.
\begin{itemize}
\item[{\rm(i)}] If $1\le k\le 4$, then $G$ is
	$k$-irreducible if and only if it is either cyclically
k-connected or the dumbbell graph.
\item[{\rm(ii)}] If $k\in\{5,6\}$, then $G$ is
	$k$-irreducible if and only if it is critical.
\item[{\rm(iii)}] If $k\ge 7$, then $G$ is $k$-irreducible
	if and only if it is bicritical.
\end{itemize}
\end{theorem}

Theorem~\ref{thm:k-irred} implies that irreducible snarks
coincide with bicritical ones, and that critical snarks are
just one step away from being irreducible. Critical snarks that
are not bicritical, called \textit{strictly critical}, appear
to be very rare. This can be observed already among snarks of
small order: there are exactly 55172 critical snarks of order
not exceeding 36, but only 846 of them are strictly critical,
just slightly over 1.5 percent \cite{BCGM, MS:crit}.

\medskip

Another important consequence of Theorem~\ref{thm:k-irred} is
that every critical snark is cyclically $4$-edge-connected and
has girth at least 5, see \cite[Proposition 4.1]{Nedela}. Thus
critical snarks are nontrivial by all generally accepted
standards.

\begin{table}[h]
	\centering
	\begin{tabular}{|c|c|c|c|c|}
		\hline
		Order & $\lambda_c = 4$ & $\lambda_c = 5$ & $\lambda_c = 6$ & Total \\
		\hline
		10 &    0 &    1 & 0 &    1\\
		18 &    2 &    0 & 0 &    2\\
		20 &    0 &    1 & 0 &    1\\
		22 &    0 &    2 & 0 &    2\\
		24 &    0 &    0 & 0 &    0\\
		26 &  103 &    8 & 0 &  111\\
		28 &   31 &    1 & 1 &   33\\
		30 &  104 &   11 & 0 &  115\\
		32 &   16 &   13 & 0 &   29\\
		34 &38827 & 1503 & 0 &40330\\
		36 &14063 &  568 & 1 &14548\\
		\hline
	\end{tabular}
	\caption{Numbers of critical snarks by
		connectivity}\label{tab:critical}
\end{table}

Table~\ref{tab:critical} indicates that among critical snarks
those with cyclic connectivity $4$ significantly prevail. It
transpires, however, that critical snarks whose cyclic
connectivity equals $4$ can be reasonably well understood
through the concept of a dot product. Given two snarks $G$ and
$H$, their \textit{dot product} $G\dotp H$ is defined as
follows (see \cite{AT, Isaacs}):  Choose two independent edges
$e=ab$ and $f=cd$ in $G$ and two adjacent vertices $u$ and $v$
in $H$. Let $a'$, $b'$ and $v$ be the neighbours of $u$, and
let $c'$, $d'$ and $u$ be the neighbours of $v$. Remove the
edges $e$ and $f$ from $G$ and the vertices $u$ and $v$ from
$H$. Finally, connect $a$ to $a'$, $b$ to $b'$, $c$ to $c'$,
and $d$ to $d'$. It is well known that $G\dotp H$ is indeed a
snark, and that it is cyclically $4$-edge-connected provided
that both $G$ and $H$ are (see \cite[Theorem~2]{AT}).

Note that the added edges $aa'$, $bb'$, $cc'$, and $dd'$ of
$G\dotp H$ form a cycle-separating $4$-edge-cut, called the
\textit{principal} $4$-edge-cut of $G\dotp H$. Thus the cyclic
connectivity of a dot product snark cannot exceed $4$. Various
authors observed that the converse holds as well: every snark
that contains a cycle-separating $4$-edge-cut can be expressed
as a dot product of two smaller snarks (see for example
\cite{Cameron, Goldberg}).

\begin{theorem}\label{thm:reverse_dp}
Every cycle-separating $4$-edge-cut $S$ in a snark $G$ gives
rise to a decomposition of $G$ into a dot product $G=G_1\dotp
G_2$ in such a way that the principal cut of $G_1\dotp G_2$
coincides with $S$. Moreover, if $G-S$ is $3$-edge-colourable,
then $G_1$ and $G_2$ are uniquely determined by $S$.
\end{theorem}

The previous theorem raises the following natural question:
What can be said about the decomposition $G=G_1\dotp G_2$ when
$G$ is critical or bicritical? The following three theorems,
proved in \cite{Chladny-Skoviera-Factorisations}, provide
answers.

\begin{theorem}\label{thm:A}
Let $G$ and $H$ be snarks different from the dumbbell graph.
Then $G\dotp H$ is critical if and only if $H$ is critical, $G$
is nearly critical, and the pair of edges of $G$ involved in
this dot product is essential in $G$.
\end{theorem}

A \textit{nearly critical snark} is one where every pair of
adjacent vertices is non-removable except possibly the pairs of
endvertices of the edges $e=ab$ and $f=cd$ of $G$ involved in
the dot product. Being \textit{essential} is a rather technical
local property which will be defined and discussed in
Section~\ref{sec:s36}.

For dot products of bicritical snarks we only have a partial
result, nevertheless, one of fundamental importance. Its
essence is the fact that the class of bicritical snarks is
closed under decompositions into a dot product.

\begin{theorem}\label{thm:B}
Let $G$ and $H$ be snarks different from the dumbbell graph. If
$G\dotp H$ is bicritical, then both $G$ and $H$ are bicritical.
Moreover, the pair of edges of $G$ involved in this dot product
is essential in $G$.
\end{theorem}

Let $G$ be a bicritical snark that contains a cycle-separating
$4$-edge-cut. By Theorems~\ref{thm:reverse_dp} and~\ref{thm:B},
we can decompose $G$ into a dot product $G=G_1\dotp G_2$ of two
smaller bicritical snarks different from the dumbbell graph. If
one of these graphs again contains a cycle-separating
$4$-edge-cut, we can continue the process. After a finite
number of steps we eventually obtain a collection
$H_1,H_2,\dots, H_r$ of cyclically $5$-edge-connected
bicritical snarks which cannot be further decomposed. Note that
the decomposition process is not uniquely determined. The
edge-cuts used on the way to a final collection of cyclically
$5$-edge-connected bicritical snarks may intersect in a very
complicated fashion, and choosing one particular 4-edge-cut at
a certain stage of the decomposition process may exclude
certain cuts from a later use in the decomposition. This
concerns especially the cuts which do not exist in the original
snark but might be, and often are, created during the process.
It is therefore rather unexpected that the following theorem is
true.

\begin{theorem}\label{thm:C}
Every bicritical snark $G$ different from the dumbbell graph
can be decomposed into a collection $\{H_1,\dots,H_n\}$ of
cyclically $5$-connected bicritical snarks such that $G$ can be
reconstructed from them by repeated dot products. Moreover,
such a collection is unique up to isomorphism and ordering of
the factors.
\end{theorem}

The assumption of Theorem~\ref{thm:C} requiring a snark $G$ to
be bicritical cannot be relaxed. Indeed, in
\cite[Section~12]{Chladny-Skoviera-Factorisations} it is shown
that there exist critical snarks with substantially different
decompositions, and even with decompositions having different
numbers of factors.

\medskip

The results discussed in the preceding paragraphs can be
summarised as follows. For every snark $G$, there exists a
sequence of edge reductions
$$G_0=G, G_1=G_0\sim e_1, \ldots, G_r=G_{r-1}\sim e_{r}$$
such that each $G_i$ is a snark and the terminal member $G_r$
of the sequence is a critical snark.

If $G_r$ is bicritical, then it has a unique decomposition into
a collection of cyclically $5$-edge-connected critical snarks
(all of which are even bicritical). Therefore $G$ can be
reconstructed from a collection of cyclically
$5$-edge-connected critical snarks by using repeated dot
products and by a series of edge extensions of snarks.

If $G_r$ is strictly critical, the situation is more
complicated, because $G_r$ can possibly be expressed as a dot
product $H_1\dotp H_2$ where $H_2$ is critical but $H_1$ is
only nearly critical. Nevertheless, as shown in
\cite[Section~6]{Chladny-Skoviera-Factorisations}, strictly
critical snarks whose cyclic connectivity equals $4$ can still
be fairly well understood. This brings us back to cyclically
$5$-edge-connected critical snarks even in the latter case.

\section{Methods of analysis}
\label{sec:methods}

Exhaustive computer search performed by Brinkmann et al.
\cite{BGHM} reveals that there are very few $5$-simple snarks with girth greater than 5 on up
to 36 vertices. These can be put aside and discussed
separately. The remaining ones have girth~5, and hence contain
a 5-cycle. If $C$ is a $5$-cycle in a snark $G$ and $e=uv$ and
$e'=u'v'$ are two edges of $C$, then according to K\'aszonyi
\cite{Kasz1} (see also Bradley
\cite{Brad}) the number of $3$-edge-colourings
of $G\sim e$ is the same as that of $G\sim e'$. In particular,
the pairs $\{u,v\}$ and $\{u',v'\}$ are either both
non-removable or both removable. It follows that connected
components of the subgraph $K$ formed by the union of all
5-cycles of $G$ are subgraphs that play a fundamental
structural role in $G$. Any such component will be called a
\emph{$5$-cycle cluster} of $G$.

The smallest $5$-cycle clusters can be found in the Petersen
graph. We call them \emph{basic} or \emph{Petersen clusters}.
It is convenient to view them as cubic multipoles with natural
partition of their semiedges into connectors, which is
determined by the way in which they were constructed from the
Petersen graph. Clearly, every $5$-cycle cluster contained in a
cyclically $5$-connected snark must have at least five
semiedges and girth $5$. To start our analysis we have therefore 
identified all $5$-cycle clusters of order up to $10$ with at least 
five semiedges. There are seven such $5$-cycle clusters, six of 
which are Petersen clusters.

\begin{itemize}

\item The \emph{pentagon} $\pentagon$ is the smallest $5$-cycle
    cluster. It consists of a single cycle of length $5$
    together with $5$ dangling edges forming its unique
    connector (see Figure~\ref{fig:pentagon}). It can be 
    constructed from the Petersen
    graph by removing any $5$-cycle.

\item The \emph{dyad} $\dyad$ (or \textit{Petersen negator}) is a
    5-cycle cluster consisting of two $5$-cycles sharing a
    path of length~$2$.  It can be constructed by removing
    a path $uwv$ of length $2$ from the Petersen graph. The
    natural distribution of semiedges into connectors makes
    it a $(2,2,1)$-pole $\dyad(I,O,R)$, with 
    $2$-connectors $I$ and $O$ containing the dangling edges 
    formerly incident with $u$ and $v$ respectively, and the 
    $1$-connector $R$ containing the only dangling edge 
    formerly incident with $w$ (see Figure~\ref{fig:dyad}). 
    The dyad has $7$~vertices, $8$ edges, and $5$ semiedges.

\item The \emph{triad} $\triad$ is a $5$-cycle cluster formed by
    three $5$-cycles $C_1$, $C_2$, and $C_3$ such that
    $C_1$ and $C_2$ have exactly one edge in common while
    $C_3$ contains the common edge of $C_1$ and $C_2$ and
    one additional edge of each $C_1$ and $C_2$. It can be
    constructed from the Petersen graph by removing one
    vertex and severing an edge not incident with it. The
    natural distribution of semiedges into connectors turns
    it into a $(2,3)$-pole $\triad(B,C)$, shown in
    Figure~\ref{fig:triad}, where the connector $B$
    corresponds to the severed edge and the connector $C$
    corresponds to the removed vertex. The triad  has $9$
    vertices, $11$ edges and $5$ semiedges.

\item The \emph{quasitriad} $\qT$ is a $5$-cycle cluster consisting of three $5$-cycles $C_1$, $C_2$, and $C_3$ such that $C_1$ and $C_2$ share two edges and $C_3$ shares one edge with each $C_1$ and $C_2$. One can simply check that performing the junction of a pair of semiedges and connecting the remaining three semiedges to a new vertex always yields a cycle of length smaller than $5$. Hence, the quasitriad is not a Petersen cluster. Like triad, quasitriad has $9$ vertices, $11$ edges, and $5$ dangling edges. The quasitriad contains exactly one pair of dangling edges at distance $1$, which distinguishes it from the triad containing two such pairs.

\item The \emph{double pentagon} $\Dp$ is a 5-cycle cluster
    containing two $5$-cycles sharing an edge. It can be
    obtained from the Petersen graph by removing two
    adjacent vertices $u$ and $v$ and severing an edge $e$
    at distance $2$ from $uv$. The natural distribution of
    semiedges turns it into a $(2,2,2)$-pole $\Dp(A,B,C)$,
    shown in Figure~\ref{fig:dp}, where the connectors $A$
    and $B$ correspond to the vertices $u$ and $v$ respectively, and $C$
    corresponds to the edge $e$. The double pentagon has
    $8$ vertices, $9$ edges and $6$ dangling edges.

\item The \emph{triple pentagon} $\Tp$ is a $5$-cycle cluster consisting of three $5$-cycles, each pair having two edges in common. It can be constructed from the Petersen graph by severing three pairwise non-adjacent edges lying on a $6$-cycle in an alternating order, making the triple pentagon  a $(2,2,2)$-pole $\Tp(A, B, C)$ as depicted in Figure~\ref{fig:tc}. Note that the three severed edges cannot be extended to a perfect matching of the Petersen graph. The triple pentagon has $10$ vertices, $12$ edges and $6$ dangling edges.

\item The \emph{tricell} $\Tc$ is a $5$-cycle cluster containing three $5$-cycles $C_1$, $C_2$ and $C_3$, where $C_1$ and $C_2$ share one edge, $C_2$ and $C_3$ share two edges, and $C_1$ and $C_3$ are disjoint. Like the triple pentagon, it arises from the Petersen graph by severing three pairwise non-adjacent edges, however, in this case these edges do not lie on a $6$-cycle and can be extended to a perfect matching of the Petersen graph.
The $3$-cell has a natural representation as a $(2,2,2)$-pole $\Tc(A, B, C)$ as shown in Figure~\ref{fig:tc}. It has $10$ vertices, $12$ edges and $6$ dangling edges. In contrast to the triple pentagon, the tricell has one dangling edge whose distance to each other dangling edge is at least~$2$.
\end{itemize}

\begin{figure}
\centering
\begin{subfigure}[c]{0.2\textwidth}
	\centering
	\includegraphics[]{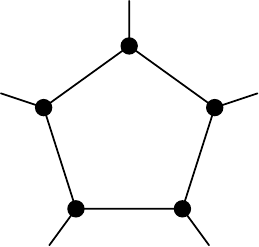}
	\caption{pentagon}
	\label{fig:pentagon}
\end{subfigure}
\begin{subfigure}[c]{0.2\textwidth}
	\centering
	\includegraphics[]{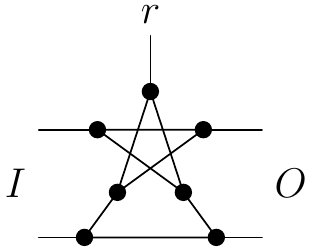}
	\caption{dyad}
	\label{fig:dyad}
\end{subfigure}
\begin{subfigure}[c]{0.25\textwidth}
	\centering
	\includegraphics[]{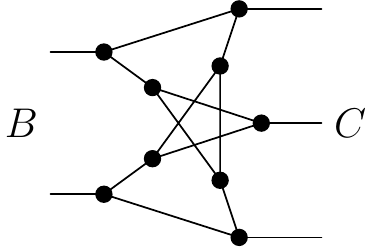}
	\caption{triad}
	\label{fig:triad}
\end{subfigure}
\begin{subfigure}[c]{0.3\textwidth}
	\centering
	\includegraphics[]{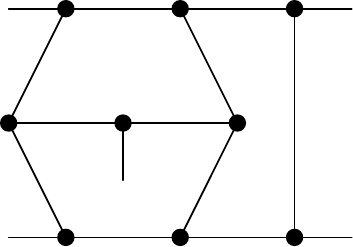}
	\caption{quasitriad}
	\label{fig:qtriad}
\end{subfigure}
\medskip

\begin{subfigure}{0.3\textwidth}
	\centering
	\includegraphics[]{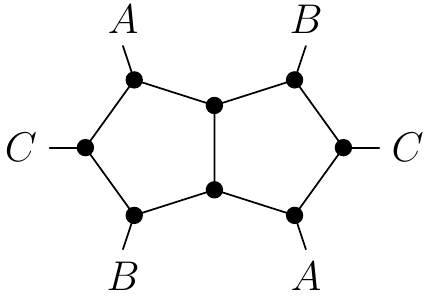}
	\caption{double pentagon}
	\label{fig:dp}
\end{subfigure}
\begin{subfigure}{0.3\textwidth}
\centering
\includegraphics[]{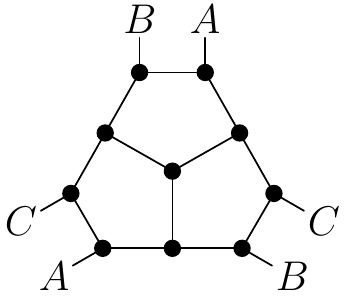}
\caption{triple pentagon}
\label{fig:tp}
\end{subfigure}
\begin{subfigure}{0.3\textwidth}
\centering
\includegraphics[]{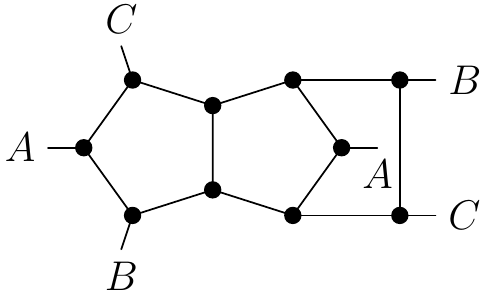}
\caption{tricell}
\label{fig:tc}
\end{subfigure}
\caption{All $5$-cycle clusters on up to $10$ vertices with girth $5$ and at least $5$ dangling edges.}
\label{fig:clusters}
\end{figure}

The key step of our analysis of $5$-simple snarks is the identification of $5$-cycle clusters.
They are easy to find by a limited-depth breadth-first search
in any given graph; we do it by employing a computer program.
Obviously, identifying the $5$-cycle clusters is not sufficient
since many snarks contain a significant number of vertices
belonging to no clusters. This is why for every given snark
we determine the structure of its $5$-cycle clusters, the
vertices not contained in clusters, and the connections between
them. This allows us to distribute almost all $5$-simple snarks into a small number of classes
depending on which basic clusters they contain and how they are
interconnected. Finally, for each class we theoretically
explain the reasons why their members are not colourable.

It is worth mentioning that quasitriads, triple pentagons, and tricells do not occur in the analysed snarks, and thus they are listed mainly for completeness. While the latter two clusters have not been observed only because the analysed critical snarks are too small, quasitriads cannot occur in critical snarks at all. Indeed, one can easily check that the quasitriad is a \emph{colour-closed $5$-pole} in the sense of \cite{Nedela}, which means that in every snark of the form  $G=\qT*M$ the complementary $5$-pole $M$ must be uncolourable. It follows that $G$ has a proper $5$-reduction, and by Theorem~\ref{thm:k-irred} such a snark cannot be critical.

\section{Commonly used multipoles}
\label{sec:multipoles}

In this section we develop tools for determining the reasons
why the analysed graphs fail to admit a $3$-edge-colouring.
Since most our graphs are built up from $5$-cycle clusters and
a number of additional vertices, the crucial point is to
analyse the colouring properties of the basic $5$-cycle
clusters and their combinations. Since our arguments only use
the fact that the Petersen graph is a snark, one can replace the
Petersen graph with a larger snark to construct a multipole in
a similar manner and with similar colouring properties as the
given basic $5$-cycle cluster. Regarding the basic $5$-cycle
clusters as special cases of these multipoles enables us to
generalise small snarks to infinite classes that cover almost
all $5$-simple snarks up to order $36$.

We start with several technical definitions that are necessary
for this purpose.

Consider a multipole $M(S_1,\dots,S_n)$ with connectors $S_1,\dots, S_n$ and a $3$-edge-colouring $\varphi$.
Define the \emph{flow through a connector $S_i$} of $M$ to be the value
$\phisum(S_i) = \sum_{e \in S_i} \phi(e)$.
Note that $\phisum(S_i)$
may happen to be $0$. A connector $S_i$ is called \emph{proper} if
$\phisum(S_i) \ne 0$ for every colouring $\phi$ of
$M$; it is called \emph{improper} if $\phisum(S_i) = 0$ for every
colouring $\phi$ of $M$. A multipole is called \emph{proper} if
all of its connectors are proper; similarly, it is \emph{improper} if all its connectors are improper.

Now let $M$ be an arbitrary $k$-pole and let $S(M)=\{e_1, e_2,
\dots, e_k\}$ be the set of its semiedges listed in a specific order given by increasing indices.  Although semiedges in connectors are generally not ordered, we fix this order only to avoid ambiguity.
For any $3$-edge-colouring $\varphi$ of $M$ and any subset $T = \{f_1, f_2, \dots, f_l\}$ of $S(M)$ we set
$\varphi(T)=(\phi(f_1),\phi(f_2), \dots, \phi(f_l))$, where the order of  semiedges in $T$  agrees with the chosen order of $S(M)$. We define the \emph{colouring set} of $M$ to be the set
$$
\col(M)=\{\phi(S(M)) \mid \phi \text{ is a Tait colouring of } M\}.
$$
Two $k$-poles $M$ and $N$ are
called \emph{colour-disjoint} if $\col(M)\cap\col(N)=\emptyset$. Note that if
$M$ and $N$ are colour-disjoint, then $M*N$ is a snark, and vice versa.

Many constructions of snarks can be conveniently described in
terms of multipole substitution. This is a very natural
concept and as such it has previously occurred in the
literature  under various disguises and different names, see
for example \cite{Fiol}.

Consider two $k$-poles $M_1$ and $M_2$.  We say that $M_1$ is
\emph{colour-contained} in $M_2$ if $\col(M_1) \subseteq
\col(M_2)$.
If $\col(M_1)=\col(M_2)$, we say that $M_1$ and $M_2$ are
called \emph{equivalent}. Now, let $G =M*N$ be a snark
expressed as a junction of two $k$-poles $M$ and $N$, and let
$M'$ be a $k$-pole colour-contained in $M$. We say that the
graph $G'=M'*N$ is obtained from the snark $G$ by a
\emph{substitution} of $M'$ for $M$. Observe that $G'$ is again
a snark: if $G'$ was colourable, then  any colouring of $G'$
could be modified to a colouring of $G$ in a straightforward
manner.

Substitution is a generic method of constructing new snarks
from old ones. It can be used in two ways: either to
produce smaller snarks with a more transparent structure, or to create larger snarks from snarks already known. Applying substitution requires having suitable pairs of multipoles. 
In the rest of this section we
describe several examples of such pairs which occur in $5$-simple
snarks. As we shall see, most of them arise from
generalisations of  the basic $5$-cycle clusters.

\subsection{Negators}

Let $G$ be a snark and let $uwv$ be a path of length two in
$G$. Removing $uwv$ from $G$ leaves a multipole $M$ whose
semiedges can be naturally distributed to two $2$-connectors
and one $1$-connector, each of them consisting of the semiedges
formerly incident with the same vertex of the path. Let
$I=\{e_1,e_2\}$ and $O=\{e_3,e_4\}$ be the connectors
consisting of the semiedges formerly incident with $u$ and $v$,
respectively, and let $R=\{e_5\}$ be the $1$-connector
containing the remaining semiedge of $M$. If $G$ is the
Petersen graph, then $M(I,O,R)$ is the dyad.

Observe that for each colouring of $M$ the total flow
through exactly one of its $2$-connectors is zero ---
otherwise the colouring could be extended to
a colouring of $G$. The flow
through the other connector is the same as that through the residual
semiedge due to the parity lemma. In other words, the colouring set
of $M$ is a subset of
$$
\mathcal{C} = \{(x,x,a,b,a+b),(a,b,x,x,a+b)\in\K^5 \mid a \ne b
\}.$$ It means that $M$ behaves like an inverting gadget which
inverts matching colours in the \emph{input connector} $I$ to
mismatching colours in the \emph{output connector} $O$, and
vice versa. For this reason $M$ is referred to as a
\emph{negator} and is denoted by $\Neg(G;u,v)$. The edge $e_5$
does not contribute to the inverting property of $M$ and
therefore it is called \emph{residual}.

The notation $\Neg(G;u,v)$ is ambiguous if $u$ and $v$ have
more than one common neighbour $w$. This can only happen if the
girth of $G$ does not exceed $4$, in which case $G$ is a
trivial snark. As we are primarily interested in nontrivial
snarks, in our work this ambiguity will be marginal.

\begin{figure}
	\centering
	\includegraphics[]{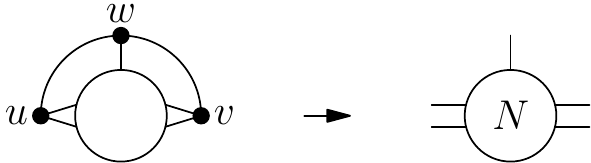}
	\caption{The snark $G$ and a symbolic representation
             of the negator $N = \Neg(G; u, v)$}
	\label{fig:neg}
\end{figure}

A negator whose colouring set is identical with $\mathcal{C}$
is called \emph{perfect}, otherwise it is called
\emph{imperfect}. For an imperfect negator $N$, it is possible
that one of its 2-connectors is improper, which means that the
other 2-connector is proper. If such a negator $N$ additionally
admits all such colourings, it is called \emph{semiperfect}.
The following theorem proved by Máčajová and Škoviera
\cite{Macajova} implies that each negator is either perfect,
semiperfect or uncolourable, and provides a characterisation of
perfect and semiperfect negators.

\begin{theorem}\label{thm:negatos} 
Let $N = \Neg(G; u, v)$ be a negator and let $w$ a common
neighbour of $u$ and $v$. If $N$ is colourable, then it is
either perfect or semiperfect. Moreover, the following hold.
\begin{enumerate}
\item[{\rm (i)}] $N$ is perfect if and only if each
of the pairs $\{u,w\}$ and $\{v,w\}$ of adjacent vertices
is nonremovable.
\item[{\rm (ii)}] $N$ is semiperfect if and only if
at least one of the pairs $\{u,w\}$ and $\{v,w\}$ is
removable.
\end{enumerate}
\end{theorem}

The smallest (connected) example of a negator can be constructed from the Petersen graph. Since the Pertersen graph graph is $2$-arc transitive, there is, up to
isomorphism, only one way of removing a path of length $2$ and thus there is, up to isomorphism, only one
Petersen negator --- the dyad. Obviously, the dyad is a perfect negator.

The parity lemma implies that every $(2,2,1)$-pole which is
colour-disjoint from a perfect negator must be a proper
$(2,2,1)$-pole. The smallest such $(2,2,1)$-pole is the path of
length $2$ with dangling edges retained and distributed into
connectors in the usual way. It consists of two endvertices
$u$ and $v$ and their common neighbour $w$. We will denote it
by $P_2(I,O,r)$, where the connector $I$ corresponds to the
dangling edges incident with~$u$, $O$ corresponds to the
dangling edges incident with $v$, and the remaining semiedge
$r$ arises from the dangling edge incident with $w$. The
colouring set of $P_2$ is
$$\col(P_2) = \{(a,b,c,d,e) \in \K^5 \mid a \ne b,\, c \ne d,\,
a+b+c+d+e = 0\}.$$

\subsection{Proper (2,3)-poles}
\label{sec:triad}

Take an arbitrary snark $G$ and choose a vertex $v$ and an edge
$e$ in $G$. Form a $(2,3)$-pole $T(B,C)$ by removing $v$ and
severing $e$. Let $B$ be the set of semiedges arising from
severing $e$ and let $C$ be the set of semiedges formerly
incident with $v$. Observe that if $G$ is the Petersen graph
and $e$ is not incident with $v$, then the result is the triad.
For connectivity reasons we only consider proper $(2,3)$-poles
that arise from a snark $G$ where the removed vertex $v$ and
edge $e$ are not incident, although the colouring properties of
$T(B,C)$ hold in the general case as well.

\begin{figure}
	\centering
	\includegraphics[]{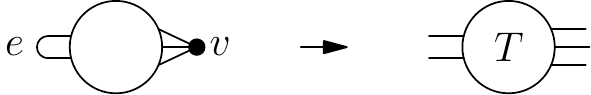}
	\caption{The snark $G$ and a symbolic representation of a proper $(2,3)$-pole $T$}
	\label{fig:proper23}
\end{figure}

We claim that $T$ is a proper $(2,3)$-pole.
Suppose not. Then $T$ admits a colouring $\varphi$ such that the flow through one of the connectors is zero. The parity lemma then implies
that $\phisum(B) = \phisum(C) = 0$ and allows extending the
colouring of $T$ to a colouring of the original snark $G$, which is impossible.
Therefore $T$ is proper. The colouring set of each proper
$(2,3)$-pole $T$ is clearly a subset of
$$\mathcal{C} = \{(a_1,a_2,b_1,b_2,b_3) \in \K^5
\mid a_1 + a_2 = b_1 + b_2 + b_3\neq 0\}.$$ If a proper
$(2,3)$-pole $T$ admits all colourings such that the flow
through each of the connectors is nonzero --- that is to say,
if its colouring set coincides with $\mathcal{C}$ --- then $T$ is called
\emph{perfect}; otherwise it is \emph{imperfect}. As one can
expect, the triad is a perfect proper $(2,3)$-pole, which can
easily be verified by hand or with the help of a computer. An
example of an imperfect proper $(2,3)$-pole will be discussed
in Section~\ref{sec:nt}.

Let us consider the multipole which is removed from a snark
$G$ in order to construct the proper $(2,3)$-pole $T(B,C)$. It is a disconnected
$(2,3)$-pole $M_{ev}(B,C)$ with connectors $B = (b_1,b_2)$ and
$C = (c_1,c_2,c_3)$ where $b_1$ and $b_2$ are the ends of an
isolated edge and $c_1$, $c_2$, and $c_3$ are ends
of three dangling edges incident with one common vertex (see also Figure~\ref{fig:ev}). Its
colouring set is
$$\col(M_{ev}) = \{(x,x,a,b,c) \in \K^5 \mid a + b + c = 0\}.$$

\subsection{Proper (3,3)-poles}

Let $u$ and $v$ be two vertices of a snark $G$. Construct a
$(3,3)$-pole $M(I, O)$ by removing $u$ and $v$ from $G$ and
putting the three semiedges formerly incident with $u$ and $v$
into the connectors $I$ and $O$, respectively. The $(3,3)$-pole
$M$ is proper: Kirchhoff's law implies that $\phisum{(I)}
=\phisum{(O)}$ for every $3$-edge-colouring $\phi$ of $M$, and
if any of these values equals zero, then $\phi$ can be extended
to a colouring of $G$, which is a contradiction. Again, for
connectivity reasons we only consider proper $(3,3)$-poles
where the vertices $u$ and $v$ are not adjacent.

In the Petersen graph there is only one way of removing a pair
of non-adjacent vertices because the graph is
$2$-arc-transitive and has diameter $2$. The result is a unique
proper $(3,3)$-pole $M_8$ of order $8$. It is easy to see
that $M_8$ can be constructed from the dyad by attaching
the residual semiedge to a new vertex incident with two
additional dangling edges, one contributing to the input and
the other one to the output.

\subsection{Even (2,2,2)-poles}
\label{sec:hexapole}

An \emph{even $(2,2,2)$-pole} is one where the number of
connectors having nonzero total flow is even. The term for this
multipole was coined by Goldberg~\cite{Goldberg} who was the
first to study this kind of multipoles. An even $(2,2,2)$-pole
can be constructed as follows. Take a snark $G$ and a vertex
$v$ with neighbours $u_1$, $u_2$ and $u_3$. Remove $v$, $u_1$,
$u_2$, and $u_3$, and for each $i \in \{1,2,3\}$ let $S_i =
\{e_i, f_i\}$ be the set containing the semiedges formerly incident with
$u_i$ that do not arise from $u_iv$. We claim that the
resulting $(2,2,2)$-pole $H(S_1, S_2, S_3)$ is even.
Kirchhoff's law tells us that it is impossible for a flow to
have a nonzero total flow through exactly one connector. It is
therefore sufficient to show that the flow through at least one
of the connectors $S_1$, $S_2$, and $S_3$ is zero. However, if
$\phisum(S_i) = c_i \ne 0$ for every $i \in \{1,2,3\}$, then
$c_1$, $c_2$, and $c_3$ are three distinct colours, so $\phi$
can be extended to a proper colouring of $G$, which is a
contradiction. Therefore $H(S_1, S_2, S_3)$ is indeed an even
$(2,2,2)$-pole.

A simple example of an even $(2,2,2)$-pole arises from a cycle 
of length six with a dangling edge at every vertex by forming 
each connector from a pair of opposite
dangling edges of the $6$-cycle. It is derived from the
Petersen graph.
Note that every even $(2,2,2)$-pole can be constructed from a
negator by deleting the residual semiedge along with its
end-vertex.

Denote by $V_4(S_1,S_2,S_3)$ the complementary $(2,2,2)$-pole
we removed from the snark $G$ to create $H$. It consists of
a vertex $v$ and its three neighbours $u_1$, $u_2$ and $u_3$,
each of them with two dangling edges attached. The connector
$S_i$ contains the semiedges belonging to the dangling edges
incident with $u_i$ (for $i \in \{1,2,3\}$). The colouring set
of $V_4$ is the set
$$\col(V_4) = \{(a_1,b_1,a_2,b_2,a_3,b_3) \in \K^6
\mid a_1 + b_1 + a_2 + b_2 + a_3 + b_3 = 0, (\forall i) (a_i \ne b_i) \}.$$

Large even $(2,2,2)$-poles can be constructed from smaller ones as follows. Let $M$ be an arbitrary $3$-pole, possibly containing loops and parallel edges, with three pairwise distinct vertices $v_1$, $v_2$, and $v_3$, each being incident with one dangling edge. Replace each vertex $v$ of $M$ with an even $(2,2,2)$-pole $H_v$ and each edge (including the dangling ones) with the $(2,2)$-pole consisting of two isolated edges which have their ends in different connectors. The result is a $(2,2,2)$-pole $H_M(S_1, S_2, S_3)$, where each connector $S_i$ consists of two dangling edges that replace the dangling edge incident with $v_i$. We claim that $H_M$ is an even $(2,2,2)$-pole. Assume that $\phisum(S_i) = a \ne 0$ for some connector $S_i$ of $H_M$. Observe that each even $(2,2,2)$-pole $H$ must have an even number of connectors with flow $a$ through it. By applying this argument inductively one can readily conclude that there exists another connector $S_j$ of $H_M$ with $\phisum(S_j) = a$. Consequently, $\phisum(S_k) = 0$ for $k \in \{1, 2, 3\} - \{i, j\}$, and therefore the number of connectors of $H_M$ with nonzero flow is either zero or two. In other words, $H_M$ is an even $(2,2,2)$-pole.

\begin{figure}[h]
	\centering
	\includegraphics{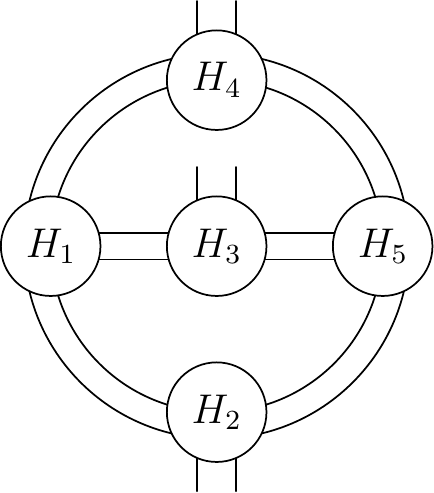}
	\caption{An even $(2,2,2)$-pole constructed from five smaller ones}
	\label{fig:odd-pole-larger}
\end{figure}

For illustration, consider the cubic graph $B_2$ consisting of two 
vertices and three parallel edges joining them, subdivide each edge 
with a new vertex, and attach a dangling edge to all three $2$-valent
vertices, thereby producing a $3$-pole $M$ on five vertices.
If we apply the construction 
described above to this $3$-pole, we obtain an even $(2,2,2)$-pole $H_M(S_1, S_2, S_3)$ composed from five smaller ones. The result is represented in Figure~\ref{fig:odd-pole-larger}, with the $(2,2,2)$-poles denoted by $H_i$ for $i\in\{1,\ldots,5\}$. Now, 
if for each of the five even $(2,2,2)$-poles we take the hexagonal $(2,2,2)$-pole derived from the Petersen graph 
and perform the junction $H_M(S_1, S_2, S_3) * V_4(S_1,S_2,S_3)$ 
we obtain the snark displayed in Figure~\ref{fig:34-6b}. Clearly, an analogous construction can be performed starting from any cubic graph in place of $B_2$. 

\begin{figure}[h!]
	\centering
	\includegraphics{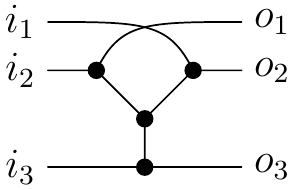}
	\caption{The Isaacs $(3,3)$-pole $Y$}
	\label{fig:y}
\end{figure}

\subsection{Isaacs $(3,3)$-poles}\label{sec:y}
Flower snarks are a well-known and in a certain sense exceptional family of snarks introduced by Isaacs in \cite{Isaacs}. Their basic building block is the 
$(3,3)$-pole $Y(I,O)$ with two ordered connectors $I = (i_1,i_2,i_3)$ and $O =(o_1,o_2,o_3)$  shown in Figure~\ref{fig:y}. It can be constructed from the complete bipartite graph $K_{3,3}$ by deleting two vertices from the same partite set, forming the connectors in the usual way, and ordering their semiedges in such a way that $\{i_1,o_2\}$, $\{i_2,o_1\}$, and $\{i_3,o_3\}$ are adjacent pairs of dangling edges.

The Isaacs flower snark $J_n$, where $n\ge 3$ is odd, can be produced by taking the disjoint union on $n$ copies $Y(I_i,O_i)$ of $Y(I,O)$, where $i\in\{1,2,\ldots,k\}$  and by joining $O_i$ to 
$I_{i+1}$ following the order of semiedges in the connectors; of course, the subscripts in the definition are taken modulo~$n$. The snark $J_3$ is the only trivial Isaacs snark.

Let $Y_k(I_1,O_k)$ denote the $(3,3)$-pole arising similarly from the union
of $k$ disjoint copies $Y(I_i,O_i)$
of $Y(I,O)$ and by performing the junction of $O_i$ and 
$I_{i+1}$ only for $i\in \{1,2,\dots,k-1\}$. It is known that 
$\col(Y_{2m}) = \col(Y_2)$ for every integer $m\ge 1$, see \cite{Nedela}.

Note that all the connectors considered in this subsection are ordered, since the arguments of the uncolourability of the Isaacs snarks require this fixed order of the semiedges in the connectors of $Y$, see for example \cite{Isaacs}.

\subsection{Proper $(2,2,1)$-poles of type NN}\label{subs:NN}

\begin{figure}[!h]
	\centering
	\begin{minipage}[c]{0.45\textwidth}
		\centering
		\includegraphics[]{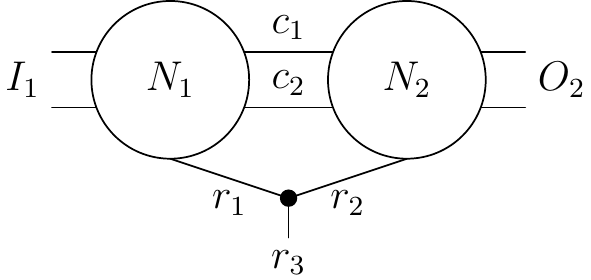}
		\captionof{figure}{A $(2,2,1)$-pole of type NN}
		\label{fig:nn}
	\end{minipage}
	\begin{minipage}[c]{0.45\textwidth}
		\centering
		\includegraphics[]{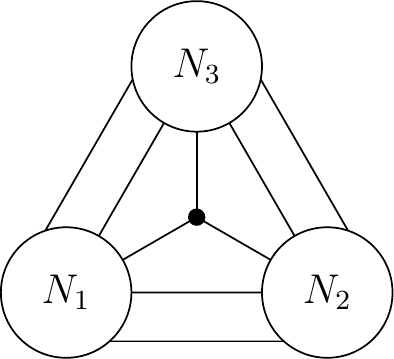}
		\captionof{figure}{A generalised Loupekine snark}
		\label{fig:nnn}
	\end{minipage}
\end{figure}

Take two
negators $N_1(I_1,O_1,r_1)$ and $N_2(I_2,O_2,r_2)$ and perform the junction of the connectors $O_1$ and
$I_2$. Add one vertex $v$ incident with the semiedges
$r_1$, $r_2$ and one new dangling edge producing a semiedge
$r_3$ (see Figure~\ref{fig:nn}). Denote the resulting
$(2,2,1)$-pole $M(I_1,O_2,r_3)$ by $\NN(N_1,N_2)$.

We prove that $\col(M)\subseteq \col(P_2)$, and if
$N_1$ and $N_2$ are perfect negators, then $\col(M) =
\col(P_2)$. Let $\phi$ be a colouring of $M$. The flow through the
connectors $O_1$ and $I_2$ has to be zero, otherwise $\phi(r_1)
= \phisum(O_1) = \phisum(I_2) = \phi(r_2)$, a contradiction.
Consequently, the flows through $I_1$ and
$O_2$ have to be nonzero. This implies that $\phi(S(M)) \in
\col(P_2)$. Now assume that $N_1$ and $N_2$ are perfect
negators. Let $\phi$ be an assignment of colours to the dangling edges of
$M$ such that $\phi(S(M)) \in \col(P_2)$, that is, $\phisum(I_1)
= a \ne 0$, $\phisum(O_2) = b \ne 0$ and $\phi(r_3) = a + b \ne
0$. We extend $\phi$ to a colouring of the multipole $M$. Set $\phi(r_1) = a$, $\phi(r_2) = b$, and $\phi(c_1) =
\phi(c_2) = a$ for the edges $c_1$ and $c_2$ connecting $N_1$ to $N_2$. Both the negators $N_1$ and $N_2$
have admissible colours on their dangling edges, and since they
are perfect, we can extend $\phi$ to the entire $(2,2,1)$-pole $M$.

Since $P_2$ is a proper $(2,2,1)$-pole, so is $M$.
Taking both the negators $N_1$ and $N_2$ from the Petersen graph --- that is, dyads --- we obtain a proper $(2,2,1)$-pole $P_{NN} = \NN(\dyad, \dyad)$ of order $15$ with $\col(P_{NN}) = \col(P_2)$.

It is worth mentioning that every snark $G$ which arises from a snark $H$ by a substitution of $\NN(N_1,N_2)$ for $P_2$ can be alternatively 
described as follows.
In the snark $G$, the complementary $5$-pole $M'$
connected to $\NN(N_1,N_2)$ is a snark $H$ with $P_2$ removed.
It means that $M'$ is the negator $N_3 = \Neg(H; u, v)$, where
$u$ and $v$ are the end vertices of~$P_2$. Thus the new snark
$G$ has the structure of a Loupekhine snark with the Petersen
negators replaced with $N_1$, $N_2$, and $N_3$ (see
Figure~\ref{fig:nnn}).

\subsection{Proper $(2,2,1)$-poles of type TT}

Take two proper $(2,3)$-poles $T_1(B_1,C_1)$ and $T_2(B_2,C_2)$
and perform the junction of $C_1$ to $C_2$. Pick one of the newly created edges,
subdivide it with a vertex $v$ and attach a dangling edge to~$v$, producing a semiedge $r$ (see Figure~\ref{fig:tt}).
Denote the $(2,2,1)$-pole
$M(B_1,B_2,r)$ constructed in this way by $\TT(T_1, T_2)$. 

\begin{figure}[!h]
	\centering
	\begin{minipage}[c]{0.45\textwidth}
		\centering
		\includegraphics[]{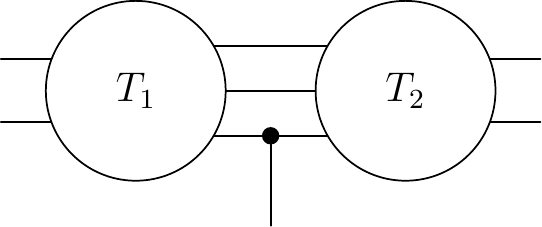}
		\captionof{figure}{A $(2,2,1)$-pole of type TT}
		\label{fig:tt}
	\end{minipage}
	\begin{minipage}[c]{0.45\textwidth}
		\centering
		\includegraphics[]{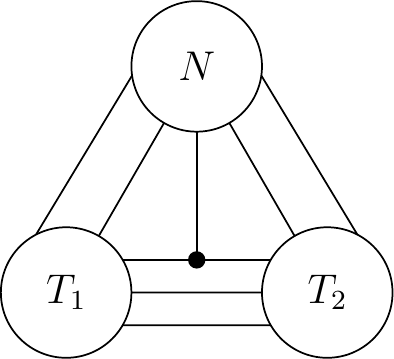}
		\captionof{figure}{Structure of a snark consisting of two proper $(2,3)$-poles $T_1$ and $T_2$ and one negator $N$}
		\label{fig:tnt-complet}
	\end{minipage}
\end{figure}

We prove that $\col(\TT(T_1,T_2)) \subseteq \col(P_2)$ with
equality attained if both $T_1$ and $T_2$ are perfect (see Section~\ref{sec:triad}).
Since $T_1$ and $T_2$ are proper, any colouring $\phi$ of
$M$ satisfies $\phisum(B_1) = a \ne 0$ and $\phisum(B_2) =
b \ne 0$, so $\phi(S(M)) \in \col(P_2)$. Assume that $T_1$
and $T_2$ are perfect. Let $\phi$ be an assignment of colours to the
dangling edges of $M$ such that $\phi(S(M)) \in \col(P_2)$. Then
there exist distinct colours $a$ and $b$ such that $\phisum(B_1)=a$,
$\phisum(B_2) = b$, and $\phi(r) = a + b \ne 0$.
Assign the edges joining $v$ to $T_1$ and $T_2$ colours $b$ and $a$, respectively. Now, if we colour the remaining two edges of $M$ joining $T_1$ to $T_2$ 
with colours $a$ and $b$ arbitrarily, we obtain admissible assignments of colours for the semiedges of both $T_1$ and $T_2$. Since both $T_1$ and $T_2$ are perfect, the assignments extend to colourings of $T_1$ and $T_2$ and hence to a colouring of the entire $M$.

Since $P_2$ is a proper $(2,2,1)$-pole, so is $\TT(T_1, T_2)$. If both proper $T_1$ and $T_2$ are obtained from the Petersen graph --- that is, if they are triads --- we get a perfect proper $(2,2,1)$-pole $P_{TT} = \TT(\triad, \triad)$ of order $19$ with $\col(P_{TT}) = \col(P_2)$.

Observe that the removal of a copy of $P_2$ from the snark $G$ yields 
a negator $N$. Hence the substitution of $\TT(T_1, T_2)$ for $P_2$ produces 
a snark consisting of two proper $(2,3)$-poles $T_1$ and $T_2$ and 
one negator $N$ as depicted in Figure~\ref{fig:tnt-complet}.

\begin{figure}[!h]
	\centering
	\begin{minipage}[c]{0.45\textwidth}
		\centering
		\includegraphics[]{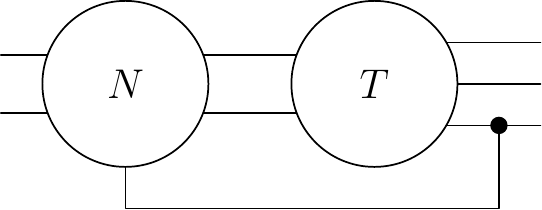}
		\captionof{figure}{A $(2,3)$-pole of type NT}
		\label{fig:nt}
	\end{minipage}
	\begin{minipage}[c]{0.45\textwidth}
		\centering
		\includegraphics[]{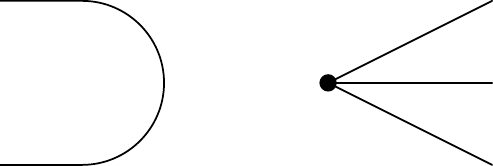}
		\captionof{figure}{The $(2,3)$-pole $M_{ev}$}
		\label{fig:ev}
	\end{minipage}
\end{figure}

\subsection{Improper $(2,3)$-poles of type NT}
\label{sec:nt}

Let $N(I,O,r)$ be a negator and $T(B,C)$ be a proper
$(2,3)$-pole. Perform the junction of $O$ and $B$, subdivide one of the dangling edges belonging to the $3$-connector $C$ of $T$, say~$e$, with a new vertex $v$, and attach the residual semiedge $r$ of $N$ to $v$ (see Figure~\ref{fig:nt}). The resulting $(2,3)$-pole $M(I,C')$ has its $2$-connector $I$ inherited from $N$ while its $3$-connector $C'$ has two semiedges $e_1$ and $e_2$ inherited from the output connector $C$ of $T$, and the third semiedge $e_3$ arises from the subdivision of $e$ with $v$. Denote this  
$(2,3)$-pole by $\NT(N,T)$.

We show that $\col(\NT(N,T)) \subseteq \col(M_{ev})$, where $M_{ev}$ is a $(2,3)$-pole shown in Figure~\ref{fig:ev} (see also Section~\ref{sec:triad}).
We also show that if both $N$ and $T$ are perfect, then $\col(\NT(N,T)) = \col(M_{ev})$.

Let $\phi$ be a colouring of the multipole $M$. Since $T$ is a
proper $(2,3)$-pole, $\phisum(B) = \phisum(O) \ne 0$ and thus
$\phisum(I) = 0$. The parity lemma implies that $\phisum(C') = 0$
and therefore $\phi(S(M)) \in \col(M_{ev})$. 

Assume that $N$
is a perfect negator and $T$ is a perfect proper $(2,3)$-pole.
Let $\phi$ be an assignment of colours to the dangling edges of $M$ such that $\phisum(I) = 0$ and $\phi(C') = (\phi(e_1),\phi(e_2), \phi(e_3)) = (a,b,c)$ where $a + b + c = 0$. To produce a colouring of $M(I,C')$, set $\phi(r) = a$, and assign the colours $b$ and $c$ to the dangling edges of $B$ in any order.
Then $\phisum(O) =\phi(r)=a=\phisum(C)$, and since $N$ and $T$ are both
perfect, this assignment extends to a colouring of $M$, as required.
Therefore $\phi(S(M)) \in \col(M_{ev})$ in this case.

As we have seen, the flow through both connectors of $\NT(N,T)$ is always zero, which means that  
$\NT(N,T)$ is an improper $(2,3)$-pole. The improper $(2,3)$-pole 
$P_{NT} = \NT(\dyad, \triad)$ is obtained by taking 
the negator and the proper $(2,3)$-pole from the Petersen graph. It has 17 vertices and $\col(P_{NT}) = \col(M_{ev})$.

If we distribute the semiedges of the connector $C'$ of the improper $(2,3)$-pole $M(I,C')$ 
into a $2$-connector and a $1$-connector, we obtain a semiperfect
negator $M'(I, C'-\{s\}, s)$ where $I$ is its improper connector and $s\in C'$ takes the role of a residual semiedge. Furthermore, if the residual semiedge $s$ of $M'$ is adjoined to $I$ to make a $3$-connector, a proper $(3,2)$-pole $M''(I\cup\{s\}, C'-\{s\})$ is obtained. Indeed, for every colouring $\phi$ of $M''$ we have 
$\phisum(I\cup\{s\})=\phisum(I)+\phi(s)=0+\phi(s)\ne 0$. However, $M''$ is an imperfect $(3,2)$-pole, because the two semiedges of $I$ always receive the same colour. 

Note that if $G = M_{ev} * T_2$ is a snark, then $T_2$ is a proper $(2,3)$-pole. Hence, the substitution of an improper $(2,3)$-pole $\NT(N, T_1)$ for $M_{ev}$ yields a snark consisting of two proper $(2,3)$-poles $T_1$ and $T_2$, and the negator $N$ as shown in Figure \ref{fig:tnt-complet}.

\subsection{Proper (2,2,2)-poles of type TTT}
\label{sec:ttt}

Take three proper $(2,3)$-poles $T_i(B_i,C_i)$, where $i \in
\{1, 2, 3\}$, and a $(3,3,3)$-pole $W=W(D_1,D_2,D_3)$ formed from a single vertex $w$ with three dangling edges by adding three isolated edges in such a way that each isolated edge contributes to different $3$-connectors of $W$. Perform the junctions $C_i$ to $D_i$ for each $i\in\{1,2,3\}$ to obtain a $(2,2,2)$-pole $M(B_1, B_2, B_3)$ shown in Figure~\ref{fig:ttt}. Denote the result by $\TTT(T_1, T_2, T_3)$.

\begin{figure}[b]
	\centering
	\begin{minipage}[c]{0.45\textwidth}
		\centering
		\includegraphics[]{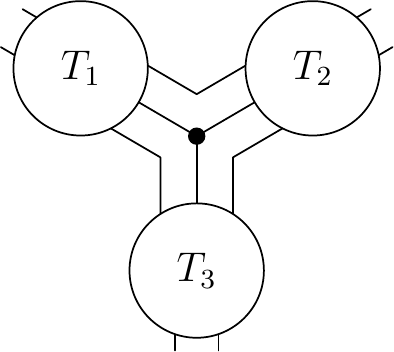}
		\captionof{figure}{A $(2,2,2)$-pole of type TTT}
		\label{fig:ttt}
	\end{minipage}
	\begin{minipage}[c]{0.45\textwidth}
		\centering
		\includegraphics[]{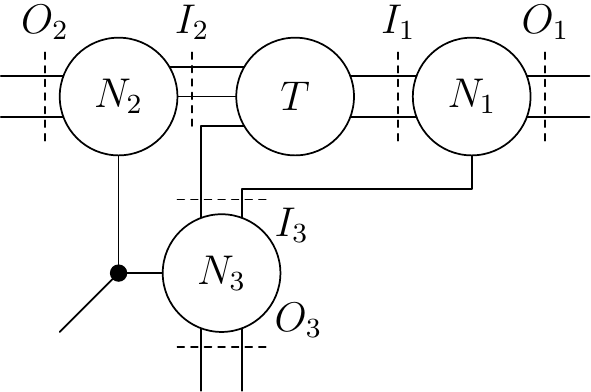}
		\captionof{figure}{A $(2,2,2,1)$-pole of type 3NT}
		\label{fig:ntnn}
	\end{minipage}
\end{figure}

We prove that
$\col(\TTT(T_1,T_2,T_3)) \subseteq \col(V_4)$, with equality
attained if all of $T_1$, $T_2$, and $T_3$ are perfect (see
Section~\ref{sec:hexapole} for the definition of $V_4$).
Let $\phi$ be a colouring of $M(B_1, B_2, B_3)$. 
Since $B_i$ is a connector of
a proper $(2,3)$-pole, we see that $\phisum(B_i) \ne 0$ for all
$i \in \{1,2,3\}$. This fact implies that $\phi(S(M)) \in
\col(V_4)$. Next, assume that $T_1$, $T_2$ and $T_3$ are
perfect. Consider an assignment $\phi$ of colours to the semiedges
of $M$ such that $c_i = \phisum(B_i)\ne 0$ for $i \in \{1,2,3\}$,  
and $c_1 + c_2 + c_3 = 0$. Let $e_1$, $e_2$, and $e_3$
be those semiedges from $C_1$, $C_2$, and $C_3$, respectively, that
are incident with the vertex $w$.
If we assign the same colour, chosen arbitrarily,
to all the remaining semiedges of the connectors $C_i$ and
set $\phi(e_i) = c_i$ for $i \in \{1,2,3\}$, we get an
admissible colouring of the semiedges for each perfect $(2,3)$-pole
$T_i$. Therefore $\phi$ can be extended to the entire multipole $M$.

We have thus showed that $\TTT(T_1, T_2, T_3)$ is a proper 
$(2,2,2)$-pole. By choosing the triads for all three proper $(2,3)$-poles we obtain a proper $(2,2,2)$-pole $P_{TTT} = \TTT(\triad, \triad, \triad)$ of order $28$ with $\col(P_{TTT}) = \col(V_4)$.

Note that after removing the $(2,2,2)$-pole $V_4$ from any snark
$G$ we obtain an even $(2,2,2)$-pole. It follows that the snark resulting from a substitution of $\TTT(T_1, T_2, T_3)$ for $V_4$ consists of an even
$(2,2,2)$-pole connected to a proper $(2,2,2)$-pole of type TTT.

\subsection{Panchromatic (2,2,2,1)-poles of type 3NT}
\label{sec:panchromatic}

Take three negators $N_i(I_i,O_i,r_i)$, for $i\in \{1, 2, 3\}$), and one proper $(2,3)$-pole $T(B,C)$. Arrange
them as depicted in Figure~\ref{fig:ntnn} and denote the
resulting $(2,2,2,1)$-pole $M(O_1,O_2,\penalty0 O_3,r)$ by
$\NNNT(N_1,N_2,N_3,T)$. 

Let $M_7(\{e_1,e_2\},I,O,r)$ denote a $(2,2,2,1)$-pole consisting of two components, an isolated edge, whose semiedges $e_1$ and $e_2$ constitute the first $2$-connector, and path of length $2$ with the standard distribution of semiedges into two $2$-connectors and a $1$-connector.
We show that $\col(M) \subseteq \col(M_7)$. Moreover, if
the negators $N_1$, $N_2$, $N_3$, and the proper
$(2,3)$-pole $T$ are all perfect, then $\col(M) = \col(M_7)$.

Let $\phi$ be a colouring of $M$, let $e$ be the edge joining $I_3$
and the $3$-connector $C$ of $T$, and let $\phi(e)=a$. Since the connector $I_1$ is joined to the proper $(2,3)$-pole $T$, there exist an element $b\in\mathbb{K}$ such that  $\phisum(I_1)=b$, whence 
$\phisum(O_1) = 0$ and $\phi(r_1) = b$.

Suppose that $a+b=\phisum(I_3)\ne 0$. Since $N_3$ is a negator, 
$\phi(r_3)=a+b$. On the other hand, the parity lemma
applied to $T$ implies that
$\phisum(I_2) =\phisum(B)+\phi(e)= a+b\ne 0$. However, 
$N_2$ is a negator, so $\phi(r_2) = a + b = \phi(r_3)$. 
Thus both $r_2$ and $r_3$ receive the same colour, which is impossible, because they are adjacent. Hence, $\phisum(I_3) = 0$ and $a=b$.

Now, $\phisum(O_3) \ne 0$ whence
$\phi(e)=b=\phi(r_1)=\phi(I_1)$. If we apply the parity lemma to $T$ again, we conclude that $b=\phisum(C)=\phisum(I_2)+\phi(e)=\phisum(I_2)+b$. Therefore $\phisum(I_2)=0$ and so 
$\phisum(O_2)\ne 0$. Summing up, $\phisum(O_1)=0$ while 
$\phisum(O_2)$, $\phisum(O_3)$, and $\phi(r)$ are all nonzero. By the parity lemma, the connectors of $M$ receive from $\phi$ all four values of $\mathbb{Z}_2\times\mathbb{Z}_2$, which justifies calling $M$ \emph{panchromatic}. Furthermore, we can immediately see that
$\phi(S(M)) \subseteq \col(M_7)$.

Now assume that all the multipoles $N_1$, $N_2$, $N_3$, and $T$ are
perfect. If $\phi$ is an assignment of colours to the dangling edges of $M$ such that $\phi(S(M)) \in \col(M_7)$, we can extend it to a
colouring of dangling edges of each of the $5$-poles $N_1$,
$N_2$, $N_3$, and $T$ in such a way that the flows through their
connectors are the same as the flows in the proof that
$\col(M) \subseteq \col(M_7)$. Such a colouring is
admissible for all of the connectors of multipoles $N_1$,
$N_2$, $N_3$, and $T$, and since they are perfect, it can be
extended to the entire $M$.

By plugging in dyads and triads we obtain a panchromatic $(2,2,2,1)$-pole $P_{3NT} = \NNNT(\dyad,\dyad,\dyad,\triad)$ of order $31$ with $\col(P_{3NT})=\col(M_7)$.

\subsection{Superpentagons}
\label{sec:superpentagon}

Let $C_5=C_5(e_0,\ldots,e_4)$ denote the 5-pole consisting of a $5$-cycle having vertices $v_0, \ldots, v_4$, arranged cyclically, with five semidges $e_0,\ldots,e_4$ attached to them correspondingly.
We define a \emph{superpentagon} to be any $5$-pole $M$ with 
$\col(M)\subseteq\col(C_5)$.  The importance of superpentagons consists in the fact that substituting a superpentagon $M$ for 
a 5-cycle $K$ in a snark $G$ produces another snark $G'$. It may be worth
mentioning that substituting a superpentagon for a $5$-cycle is equivalent 
to a $5$-product of $G$ with a snark $\bar M$ obtained from $M$ by joining
$M$ to a $5$-cycle in a Petersen-like manner (that is, as a pentagram). 
For the definition of a $5$-product of snarks see
\cite[last paragraph of Section~3]{Cameron} or \cite[pp.~51-52]{Chladny-Skoviera-Factorisations}.

It is a well known fact, proved in \cite[Lemmas~6.2-6.5]{Nedela}, that for an arbitrary $5$-pole $M$ with $\col(M)\subseteq\col(C_5)$ only two possibilities can occur: either $\col(M)=\emptyset$  or $\col(M)=\col(C_5)$. In the latter case we call $M$ a \emph{perfect superpentagon}.
A familiar example of a perfect superpentagon has 15 vertices and can be obtained from the Isaacs flower snark $J_5$ by removing the unique $5$-cycle $C$ of $J_5$ and changing the cyclic order $e_0,e_1,e_2,e_3,e_4$ of the resulting dangling edges to $e_0,e_2,e_4, e_1,e_3$.

We now describe a superpentagon $Q(f_0,\ldots,f_4)$ that 
can be observed in the analysed graphs. 
Let $T=T(B,C)$ be a proper $(2,3)$-pole with  
$B = \{d_1,d_2\}$.
By distributing the semiedges of $B$ into two $1$-connectors we obtain a $(3,1,1)$-pole $T(C, d_1, d_2)$, which is proper as well. 
We further need a proper $(3,3)$-pole $R=R(I,O)$ and a 
$(3,3,1)$-pole $U=U(S_1,S_2,r)$ with 
$S_i = \{e_i,f_i,g_i\}$ and $i\in \{1,2\}$. It
consists of one vertex $v$ incident with three dangling edges
whose semiedges are $e_1$, $e_2$ and $r$, and two isolated
edges with semiedges $f_1$, $f_2$ and $g_1$, $g_2$, respectively (see $U_2$ and $U_3$ in Figure~\ref{fig:sup5}). 

To construct $Q(f_0,\ldots,f_4)$ proceed as follows (see Figure~\ref{fig:sup5}):
\begin{itemize}
\item Set $f_0=e_0$.

\item For $i\in\{1,4\}$ substitute a copy $T_i$ of $T(C, d_1, d_2)$ for 
the vertex $v_i$ of $C_5$. Define $f_i$ to be the copy of the semiedge 
$d_1$ in $T_i$. 

\item For $i\in\{2,3\}$ substitute a copy $U_i$ of $U(S_1,S_2,r)$ for the vertex $v_i$ of $C_5$. Define $f_i$ to be the copy of the semiedge $r$ in $U_i$.

\item Substitute the proper $(3,3)$-pole $R=R(I,O)$the for the  edge 
$v_2v_3$ of $C_5$. Join $I$ to the copy of the connector $S_2$ in $U_2$, and further join $O$ to the copy of $S_1$ in $U_3$.

\item Join the copy of $S_1$ in $U_2$ to the $3$-connector of $T_1$. Join the copy of $S_2$ in $U_3$ to the $3$-connector of $T_4$.

\item Attach the copy of $d_2$ in both $T_1$ and $T_4$ to $v_0$.
\end{itemize}

It is easy to check that $Q(f_0,\ldots,f_4)$ is a superpentagon. Moreover, if the involved $(3,3)$-pole $R$ is colourable and both proper $(2,3)$-poles corresponding to $T_1$ and $T_4$ are perfect, then $Q$ is a perfect superpentagon. 

The smallest example that can be produced by this
method has 29 vertices. It is constructed by 
choosing $T$ to be the triad and $R$ to be the proper 
$(3,3)$-pole on eight vertices. Clearly, this superpentagon is proper.

\begin{figure}
	\centering
	\includegraphics[]{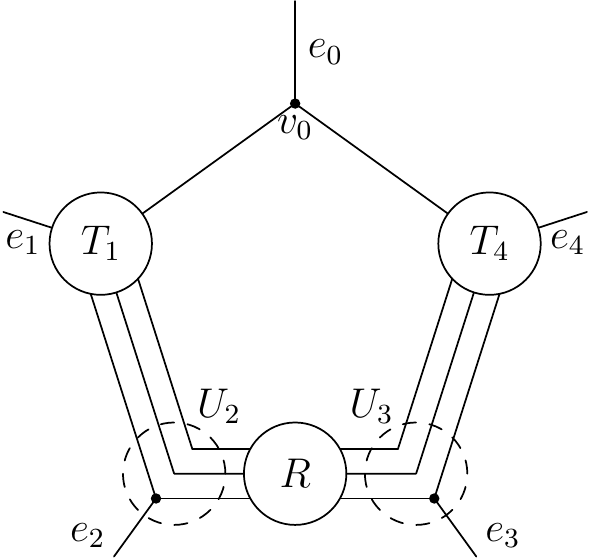}
	\caption{A superpentagon}
	\label{fig:sup5}
\end{figure}

\section{Results of analysis}
\label{sec:results}

Having specified all necessary tools, we present results of our analysis. In the present section we describe only bicritical snarks since critical snarks that are not bicritical deserve special attention, which is why they will be treated in a separate section.

For a given snark $G$ we start by identifying all $5$-cycle clusters as explained in Section~\ref{sec:methods}. Subsequently we
check whether $G$ contains any of the $5$-poles $P_{NN} = \NN(\dyad, \dyad)$, $P_{NT} = \NT(\dyad, \triad)$, $P_{TT} = \TT(\triad, \triad)$, $P_{TTT} =
\TTT(\triad, \triad, \triad)$, and $P_{3NT} = \NNNT(\dyad,\dyad,\dyad,\triad)$ which we have described in
Section~\ref{sec:multipoles}.
As we have proved in Section~\ref{sec:multipoles}, the
multipoles $P_{NN}$, $P_{NT}$, $P_{TT}$, $P_{TTT}$, and $P_{3NT}$ are
colour-equivalent to $P_2$, $M_{ev}$, $P_2$, $V_4$, and $M_{7}$,
respectively. Here we make use of the fact that both $\dyad$ and
$\triad$ are perfect. Whenever a snark $G$ contains any of those
multipoles, it can be constructed from a smaller snark $G'$ by a suitable substitution. 
The order of $G'$
is easy to compute: if $|G|\le 36$, then $|G'|\le 24$. There
are not so many snarks of order at most $24$, thus it is
easy to check if $G'$ is isomorphic to one of them. It may be useful to note that in all cases where this procedure could be applied, the snark $G'$ was nontrivial, although sometimes its cyclic connectivity was smaller  (necessarily $4$) or $G'$ was reducible.

If we discard the snarks arising from smaller snarks by one of
the substitutions mentioned above, some snarks will remain. We distribute them, for each particular order, into
several classes.
Each class can be characterised by specific
junctions of suitable multipoles. Most of these multipoles are just
$5$-cycle clusters,
but several more interesting ones have
emerged, too (for example, $M_{11}$ in Figure~\ref{fig:32}). In
contrast to $5$-cycle clusters, those have been analysed by
hand. We present each class as an infinite family of snarks containing the desired small snarks.

Up to isomorphism, there is only one Petersen negator $N_P$, the dyad 
$\dyad$, and only one Petersen proper $(2,3)$-pole $T_P$, the triad $\triad$. However, when performing a
junction of two connectors containing more than one semiedge,
the result may depend on the particular order of semiedges in the
connectors. Typically, the order does not affect our uncolourability
arguments because they involve the connectors in their
entirety, the only exception being the Isaacs snarks. Nevertheless, choosing different order of semiedges in connectors may
lead to several non-isomporphic variations of the multipoles
$P_{NN}$, $P_{NT}$, $P_{TT}$, and $P_{TTT}$. It would be
possible to take this into account in our classification, but
such level of detail would just obscure the analysis without
tangible benefits. Perhaps the only case where one might
be interested in distinguishing those non-isomorphic variants
occurs when using them as construction blocks in order to obtain
larger snarks with specific properties.

By applying the approach mentioned above we have analysed 
all cyclically $5$-connected bicritical snarks
with at most 36 vertices. The list of such snarks was obtained
from \cite{BCGM}. The results are summarised in
Table~\ref{tab:all}. Uncolourability of certain snarks can be
explained in several different ways. Consequently, they are
included in more than one of our classes, which explains why
the numbers in Table~\ref{tab:all} do not add up.

\begin{table}[h]
	\centering
	\begin{tabular}{|r|ccccccccccc|}
		\hline
		Order & 10 & 20 & 22 & 24 & 26 & 28 & 30 & 32 & 34 & 36 & 38  \\
		\hline
NN substitution  & 0  & 0  & 2  & 0  & 0  & 0  & 10 & 11 & 26 & 10 & $\ge 39$  \\
TT substitution  & 0  & 0  & 0  & 0  & 8  & 0  & 0  & 0  & 84 & 69 & $\ge 3$  \\
NT substitution  & 0  & 0  & 0  & 0  & 8  & 0  & 0  & 0  &1084& 396& $\ge 17$ \\
TTT substitution & 0  & 0  & 0  & 0  & 0  & 0  & 0  & 0  & 22 & 0  & $\ge 0$ \\
Superpentagon subst. & 0  & 0  & 0  & 0  & 0  & 0  & 0  & 0  & 72 & 0  & $\ge 0$ \\
Other         & 1  & 1  & 0  & 0  & 0  & 1  & 1  & 2  &215 & 9  & $\ge 4$ \\
		\hline
TOTAL         & 1  & 1  & 2  & 0  & 8  & 1  & 11 & 13 &1503& 484& $\ge 56$ \\
		\hline
	\end{tabular}
	\caption{Classification of all cyclically $5$-connected bicritical snarks of order up to $36$}
	\label{tab:all}
\end{table}

\subsection*{Isaacs' flower snarks}

We have already defined the Isaacs flower snarks $J_n$, where $n\ge 3$ is odd, in Section~\ref{sec:y}. An alternative approach to describing them uses substitutions starting from the Petersen graph.
The flower snark $J_3$ arises from the Petersen graph by substituting 
a triangle for a vertex. For each $k \ge 2$, the snark $J_{2k+1}$ can 
be constructed from $J_{2k-1}$ by substituting the $(3,3)$-pole $Y_4$ for a copy of the $(3,3)$-pole
$Y_2$ contained in it, both $(3,3)$-poles having the same colouring set 
(for the definition of $Y_i$ see Section~\ref{sec:y}).
It is well known that the flower snarks $J_n$ are cyclically $6$-edge-connected for all $n \ge 7$ and bicritical for $n\ge 5$ (see \cite[Proposition~4.7]{Nedela}).

There are three Issacs flower snarks up to order $36$, 
namely $J_5$, $J_7$, and $J_9$.
The snark $J_5$ with $20$ vertices contains single cycle-separating
$5$-cut; the next one, $J_7$ with $28$ vertices, is the
smallest cyclically $6$-connected snark, and $J_9$ is also cyclically $6$-connected. The Isaacs snarks will not
be mentioned anymore in the rest of this classification.

\subsection*{Order 22}

The only cyclically $5$-connected bicritical snarks of order
$22$ are the two Loupekine snarks (shown in Figure~\ref{fig:nnn}).
Both of them contain the $5$-pole $P_{NN}$, so they can be
constructed from the Petersen graph by a substitution of $P_{NN}$
for $P_2$.

\subsection*{Order 26}

There are eight cyclically $5$-connected bicritical snarks of
order $26$. All of them contain a proper $(2,2,1)$-pole
$P_{NT}$ and also a proper $(2,2,1)$-pole $P_{TT}$. Hence they
can be obtained by a substitution from the Petersen graph.
Also, each of them is spanned by the uncolourable $7$-pole from
Figure~\ref{fig:tnt} consisting of one dyad $N$ connected to two
triads $T_1$ and $T_2$. This multipole was used by Steffen and others
\cite{Oddness, SteffenClassifications} to construct cyclically
$5$-connected snarks with small order and large resistance.

\begin{figure}[h!]
        \centering
        \includegraphics[]{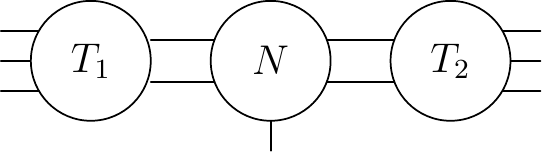}
        \captionof{figure}{An uncolourable $7$-pole}
        \label{fig:tnt}
\end{figure}

\subsection*{Order 30}

On $30$ vertices, there is one cyclically $5$-connected snark of girth $6$, the double-star snark discovered by Isaacs \cite{Isaacs}. It can be described as a $5$-product of $J_5$ with itself. All the
remaining snarks of order $30$ arise from the Blanuša snarks by
substituting $P_{NN}$ for $P_2$: six snarks from the Type~1 Blanu\v sa snark
and four snarks from the Type~2. By the Type~1  Blanuša snark we mean a nontrivial snark of order~18 discovered in 1946 by Blanu\v sa \cite{Blanusa},  and by the Type~2 Blanu\v sa snark we mean the other nontrivial snark on 18 vertices, the Blanu\v sa double as it is called in \cite{Orbanic}, where the history and properties of these two snarks are discussed in detail. 
Note that this substitution increases cyclic connectivity.

\subsection*{Order 32}

There are 13 cyclically $5$-connected bicritical snarks of
order 32. From among them, 11~contain the $5$-pole $P_{NN}$.
All of them can be constructed from the flower snark $J_5$ by a
substitution. The remaining two constitute Class 32-A which is
in detail described below.

\begin{figure}[h!]
	\centering
	\begin{minipage}[c]{0.40\textwidth}
		\centering
	\includegraphics[]{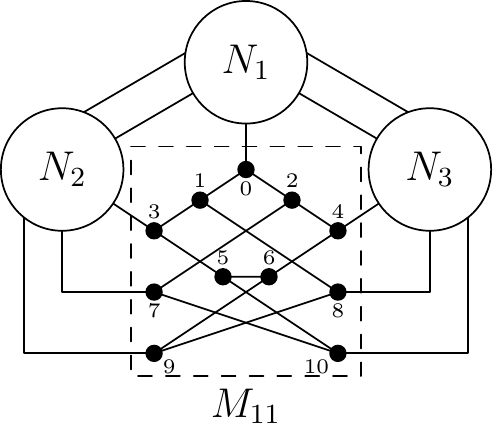}
		\captionof{figure}{The structure of Class 32-A}
		\label{fig:32}
	\end{minipage}
	\hfill
	\begin{minipage}[c]{0.40\textwidth}
		\centering
		\includegraphics[]{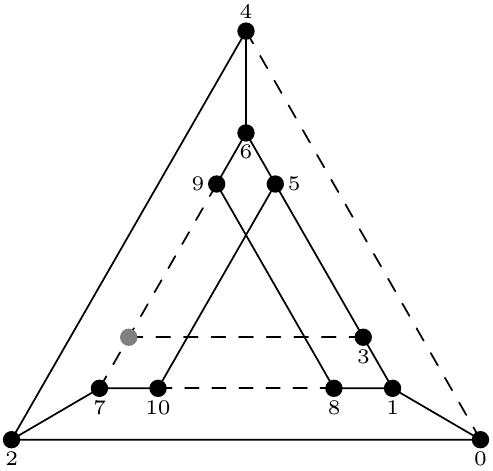}
		\captionof{figure}{The $7$-pole $M_{11}$ constructed from $J_3$}
		\label{fig:32-J3}
	\end{minipage}
\end{figure}

\subsubsection*{Class 32-A}

The two snarks of this class consist of three Petersen negators
$N_i(I_i,O_i;r_i)$ for  $i \in \{1,2,3\}$ and one $7$-pole
$M_{11}$, which are combined as shown in Figure~\ref{fig:32}. In
general, the negators $N_1$, $N_2$, and $N_3$ can be taken from an
arbitrary snark. The $7$-pole $M_{11}$ can be derived from the
flower snark $J_3$ by removing one vertex and severing two edges
as indicated in Figure~\ref{fig:32-J3}.

We now explain why any graph $G$ represented by Figure~\ref{fig:32} 
is a snark.
By contradiction, suppose that $G$ has a $3$-edge-colouring $\phi$.
To avoid ambiguity assume that, in $G$, the connector $O_2$ is joined 
to $I_1$  and the connector $O_1$ is joined to $I_3$. Let $a = \phi(r_1)$.
One of the connectors of $N_1$, say, $I_1$, must have zero flow.
Then $\phisum(O_1) = a = \phisum(I_3)$. From the negator $N_3$,
we get that $\phisum(O_3) = 0$ and $\phi(r_3) = a$. Now, consider the $7$-pole $M_{11}$. The semiedges $e_1$ and $e_2$
connected to the connector $O_3$ have the same colour and so
have the semiedges $e_3$ and $e_4$ connected to $r_3$ and
$r_1$, respectively. The parity lemma implies that the sum of the flows
through the remaining three semiedges $e_5$, $e_6$, and $e_7$ is
zero. Therefore, we can perform junctions of $e_1$ with $e_2$,
$e_3$ with $e_4$ and add one new vertex incident with $e_5$,
$e_6$ and $e_7$, giving rise to a graph $H$ with edges coloured
by $\phi$. However, the graph $H$ is isomorphic to the flower
snark $J_3$, which is a contradiction.

Note that we require that $M_{11}$ can be extended to a snark in
two symmetric ways, so we cannot replace it with an arbitrary
$7$-pole constructed from a snark by removing a vertex and
cutting two edges.

\subsection*{Order 34}

Among the $1503$ bicritical cyclically $5$-connected snarks of order $34$, we
have found $26$ that contain the $5$-pole $P_{NN}$. They can be constructed from the Loupekine snarks (both types) by a substitution for $P_2$. The $5$-pole $P_{NT}$ is contained in $1084$ snarks
which arise from the Blanuša snark (both types) by a
substitution for one vertex and one edge. The next $84$ snarks
are spanned by the $5$-pole $P_{TT}$ and can be constructed from the
Blanuša snark (both types) by a substitution of $P_2$.
Further $72$ snarks arise from the Petersen graph by substituting  a superpentagon described in Section~\ref{sec:superpentagon} for a pentagon.

After analysing the structure of the remaining snarks we
have decided to categorise them into six classes. 
The classification for order $34$ is summarised in
Table~\ref{tab:34}.

\begin{table}[h]
	\centering
	\begin{tabular}{|r|l|}
		\hline
		Type of a snark & Number of snarks \\
		\hline
		Containing $P_{NN}$  & $26$ \\
		Containing $P_{NT}$  & $1084$\\
		Containing $P_{TT}$  & $84$ \\
		Containing $P_{TTT}$ & $22$ \\
		Containing superpentagon & $72$ \\
		Class 34-A           & $21$ \\
		Class 34-B           & $18 + 18\ (P_{TT}) + 54\ (P_{NT}) = 90$ \\
		Class 34-C           & $162 + 18\ (P_{NT}) = 180$ \\
		Class 34-D           & $5$ \\
		Class 34-E           & $7$ \\
		Class 34-F           & $2$ \\
		\hline
		TOTAL                & $1503$\\
		\hline
	\end{tabular}
	\caption{Structure of cyclically $5$-connected bicritical snarks of
order $34$.}
	\label{tab:34}
\end{table}

\subsubsection*{Class 34-A}

\begin{figure}[t]
	\centering
	\begin{minipage}[c]{0.45\textwidth}
        \centering
		\includegraphics[]{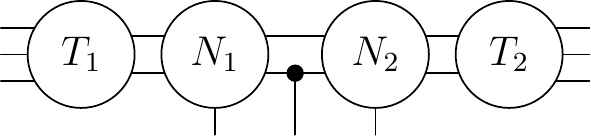}
		\captionof{figure}{Uncolourable $9$-pole $M_1$ (Class 34-A)}
		\label{fig:34-A}
	\end{minipage}
	\hfill
	\begin{minipage}[c]{0.45\textwidth}
		\centering
		\includegraphics[]{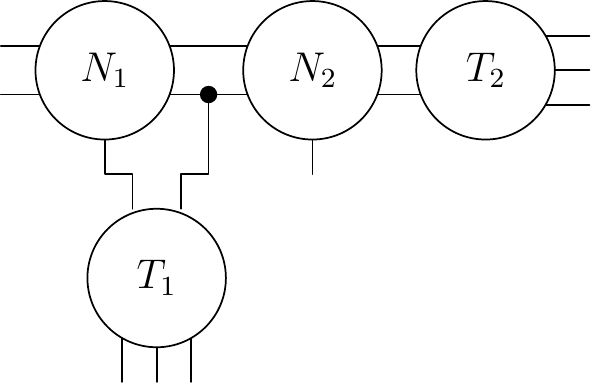}
		\captionof{figure}{Uncolourable $9$-pole $M_2$ (Class 34-B)}
		\label{fig:34-B}
	\end{minipage}
\end{figure}

Take two negators $N_1$ and $N_2$ and two proper $(2,3)$-poles
$T_1$ and $T_2$, and construct a $9$-pole $M_1$ as shown in
Figure~\ref{fig:34-A}. We prove that $M_1$ is uncolourable. 
Suppose to the contrary that $M_1$ admits a colouring $\phi$. 
Let $e$ denote the dangling edge
which is incident with none of the multipoles $N_1$, $N_2$,
$T_1$, and $T_2$. If a connector of a negator is joined to a
proper connector, then its other connector must have zero flow
through it. In our case, this is true for both $N_1$ and $N_2$.
After applying the parity lemma to the $(2,2,1)$-pole connecting
$N_1$ to $N_2$, which contains one vertex with three dangling
edges and one isolated edge, we obtain $\phi(e) = 0$, which is a
contradiction.

Among the studied snarks of order $34$, there are $21$ snarks
containing $M_1$. In all of them, $M_1$ is built from dyads and
triads and has $33$ vertices.

\subsubsection*{Class 34-B}

Assume that we have two negators $N_1$ and $N_2$ and two proper
$(2,3)$-poles $T_1$ and $T_2$ which are assembled to a $9$-pole
$M_2$ as depicted in Figure~\ref{fig:34-B}. We prove that $M_2$
is uncolourable. Let $v$ denote the vertex of $M_2$ not belonging
to any of $N_1$, $N_2$, $T_1$, and $T_2$, and let $e$ be the edge between
$v$ and $T_1$. Suppose to the contrary that there is a colouring 
$\phi$ of $M_2$. Since a
connector of $N_2$ is joined with a proper connector, the flow
through its other connector is zero. If we apply the parity lemma
first to the $(2,2,1)$-pole containing $v$ (which is joined
to $N_1$, $N_2$, and $T_1$) and then to $N_1$, we can conclude that 
the flow through $e$
is the same as the flow through the residual semiedge of $N_1$.
These values force a zero flow through a connector of $T_1$, which is
impossible.

The multipole $M_2$ contained in the studied snarks of order $34$
consists of dyads and triads and has $33$ vertices. We have
identified $90$ snarks containing this $9$-pole. Of them, $54$
snarks contain also the $5$-pole $P_{NT}$, and $18$ snarks
contain the $5$-pole $P_{TT}$.

\subsubsection*{Class 34-C}

Let $G$ be a snark. Delete a path $uv$ from $G$, sever an edge $e \ne uv$ of $G$, and denote the resulting $(2,2,2)$-pole by
$R(A,B,C)$, where the connectors $A$ and $B$ contain the two semiedges formerly incident with $u$ and $v$, and $C$ contains two semiedges that arose from severing $e$.
Should $R$ be contained in a cyclically $5$-connected graph, $e$ cannot be incident with any of $u$ and $v$. 
The crucial property of $R$ is that it admits no
colouring $\phi$ such that $\phisum(A)\neq 0$ and $\phisum(C) =
0$. Clearly, any such colouring could be extended to the entire snark $G$. Observe that if we choose the Petersen graph for $G$, we obtain the double pentagon as the $(2,2,2)$-pole $R$.

Take the $(2,2,2)$-pole $R(A, B, C)$, a negator
$N(I,O)$, and two proper $(2,3)$-poles $T_1(B_1,C_1)$ and
$T_2(B_2,C_2)$. Join $B_1$ to $A$, $C$ to $I$, $O$ to
$B_2$, and denote the resulting
$9$-pole $M_3$ (see Figure~\ref{fig:34-3}). Assume that $M_3$ has a colouring $\phi$. The
negator $N$ is connected to the proper $(2,3)$-pole $T_2$,
hence $\phisum(I) = \phisum(C) = 0$. Since  $T_1$ is also
proper, we have $\phisum(B_1) = \phisum(A) \neq 0$. This
colouring is impossible for $R$---a contradiction.

\begin{figure}[h]
	\centering
	\includegraphics[]{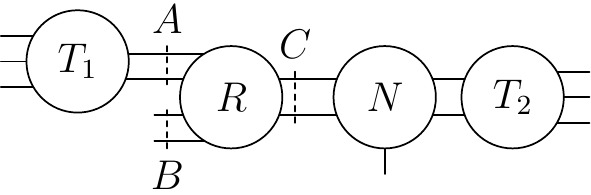}
	\caption{Uncolourable $9$-pole $M_3$ (Class 34-C)}
	\label{fig:34-3}
\end{figure}

There are $72$ bicritical $5$-connected snarks of order $34$
belonging to this class ($18$ of them also contain $P_{NT}$).
In all of them, the negators and proper $(2,3)$-poles are
derived from the Petersen graph. 

\subsubsection*{Class 34-D}

\begin{figure}[t]
	\centering
	\begin{minipage}[c]{0.45\textwidth}
		\centering
		\includegraphics[]{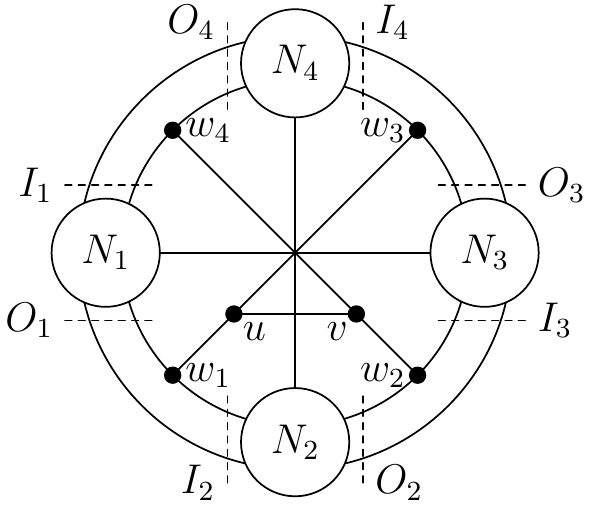}
		\captionof{figure}{The structure of Class 34-D snarks}
		\label{fig:34-4}
	\end{minipage}
	\hfill
	\begin{minipage}[c]{0.45\textwidth}
		\centering
		\includegraphics[]{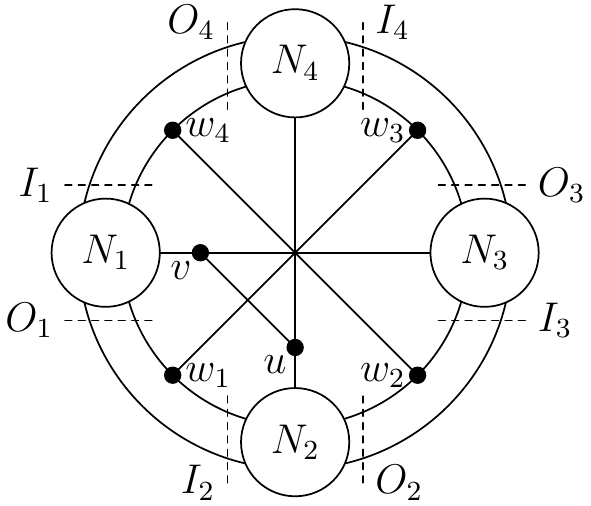}
		\captionof{figure}{The structure of Class 34-E snarks}
		\label{fig:34-5}
	\end{minipage}
\end{figure}

Take four negators $N_i(I_i, O_i, r_i)$ for $i \in
\{1,2,3,4\}$, connect them as shown in Figure~\ref{fig:34-4}
and denote the resulting graph by $G_{4}$. We prove that
$G_{4}$ is a snark. If it is not, then it has a colouring
$\phi$. Without loss of generality we may assume that $\phisum(O_1)
= 0$. Let $\phi(uw_1) = a \in \K$; then $\phisum(I_2) =
\phisum(O_1) + \phi(uw_1) = a \ne 0$, so $\phisum(O_2)=0$.
Analogously, we get $\phisum(O_3) = \phisum(O_4)=0$. By
repeatedly applying the parity lemma we get
$\phi(w_3u)=\phisum(I_4)=\phi(r_4)=\phi(r_2)=\phisum(I_2)=a$,
thus there is a colour conflict at $u$.

We have identified six snarks having this structure, with all
four negators taken from the Petersen graph. All of them are permutation snarks, which means that they admit a $2$-factor consisting of two induced cycles. For more information about permutation snarks one can consult
\cite{BGHM, MS:perm-short, MS:superp}.

\subsubsection*{Class 34-E}

The four negators $N_1$, $N_2$, $N_3$ and $N_4$ can also be
arranged in a different way shown in Figure~\ref{fig:34-5}. The
proof of uncolourability is similar to the one for the
previous class. There are six snarks of order $34$ having this
structure. The $12$ snarks constituting the classes 34-D and 34-E form a complete set of all cyclically $5$-connected permutation snarks of order $34$, see Brinkmann et al. \cite{BGHM}.

\subsubsection*{Class 34-F}

We take five even $(2,2,2)$-poles $H_1$, $H_2$, $H_3$, $H_4$
and $H_5$, and from them we construct a larger $(2,2,2)$-pole
$H_M$ as explained in Section \ref{sec:hexapole}. Since $H_M$ is an
even $(2,2,2)$-pole, $H_M*V_4$ is a snark. Class 34-F consists of
snarks of the form $H_M*V_4$ whose scheme can be seen in
Figure~\ref{fig:34-6c}. Since each even $(2,2,2)$-pole $H_i$, for $i \in \{2, 3, 4\}$, has its semiedges from one of its connectors
joined to a vertex, which produces a negator $N_i$, the structure of Class 34-F snarks can be described using two even $(2,2,2)$-poles $H_1$ and $H_5$ and three negators $N_2$, $N_3$ and $N_4$ as depicted in Figure~\ref{fig:34-6a}.

Among the studied snarks of order $34$, there are two snarks of
this class, both using even $(2,2,2)$-poles derived from the
Petersen graph (i.e. hexagons); or alternatively, two hexagons and three dyads.

\begin{figure}[b]
	\centering
	\begin{minipage}[c]{0.45\textwidth}
		\centering
		\includegraphics[]{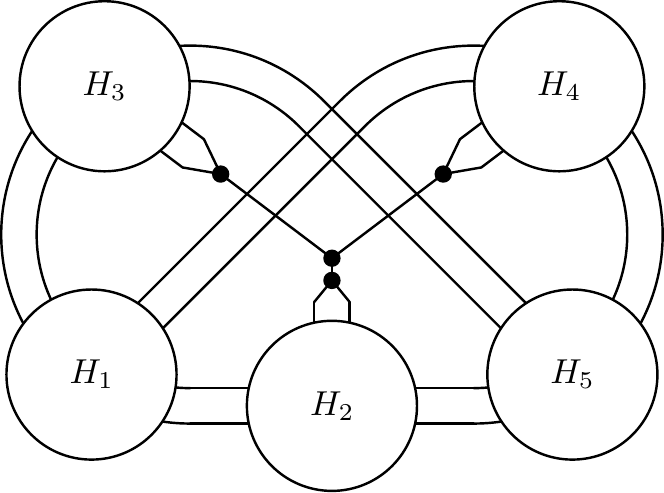}
		\captionof{figure}{The structure of Class 34-F snarks}
		\label{fig:34-6c}
	\end{minipage}
	\hfill
	\begin{minipage}[c]{0.45\textwidth}
		\centering
		\includegraphics[]{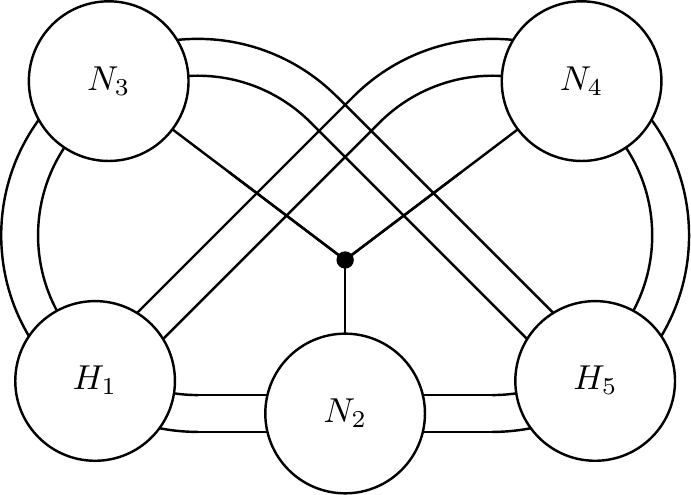}
		\captionof{figure}{Negators in Class 34-F snarks}
		\label{fig:34-6a}
	\end{minipage}
\end{figure}

\subsection*{Order 36}

Out of the $484$ bicrititical
cyclically $5$-connected snarks of this order, $396$ snarks
contain the $5$-pole $P_{NT}$ and all of them arise from the
flower snark $J_5$ by a substitution for $P_2$. In $69$ snarks we
have identified the $5$-pole $P_{TT}$ and all of them arise from the
$J_5$ by a substitution for $P_2$.

The $5$-pole $P_{NN}$ is contained in ten snarks. They are
constructed by substitution from two smaller snarks of
order $24$ whose structure can be seen in Figure~\ref{fig:24}.
Their structure is similar to Loupekine snarks; they only
contain two additional vertices $u$ and $v$ joined by an edge. The pair of $\{u,v\}$ is removable, so these snarks are reducible, although after a
substitution of $P_{NN}$ for a path $P_2$ the resulting snarks become
irreducible. In all ten cases, the $(2,2)$-pole $P_2$ used for the substitution contains one of the vertices $u$ or $v$.

\begin{figure}
	\centering
	\includegraphics[]{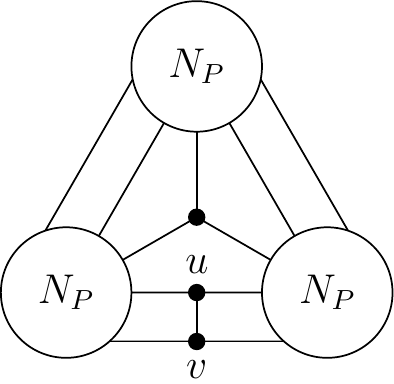}
	\caption{A scheme for two cyclically $5$-connected reducible snarks of
order $24$.}
	\label{fig:24}
\end{figure}

Excluding the flower snark $J_9$, there remain eight snarks of order 36, which fall into two classes.

\subsubsection*{Class 36-A}

Take three negators $N_i(I_i,O_i,r_i)$, where $i\in\{1,2,3\}$, with unordered connectors (as usual), arrange them cyclically and join each $O_i$ to $I_{i+1}$, where the indices are reduced modulo $3$. Subdivide one of the two edges connecting $N_i$ to $N_{i+1}$ with a new vertex $v_i$ and attach a dangling edge $e_i$ to $v_i$, thereby producing a cubic $6$-pole. Turn it into a $(3,3)$-pole $M_{24}(I,O)$ by distributing the semiedges into two ordered connectors $I = (r_1,r_2,r_3)$ and $O = (e_3,e_1,e_2)$; see Figure~\ref{fig:36-A}.

Consider the $(3,3)$-pole $Y_{3}(I_Y,O_Y)$
consisting of three copies of the Isaacs $(3,3)$-pole~$Y$, defined in
Section~\ref{sec:y}. We prove that $M_{24}$ and $Y_3$ are
colour-disjoint, implying that $M_{24}*Y_3$ is a snark.

\begin{figure}[b]
	\centering
	\includegraphics[]{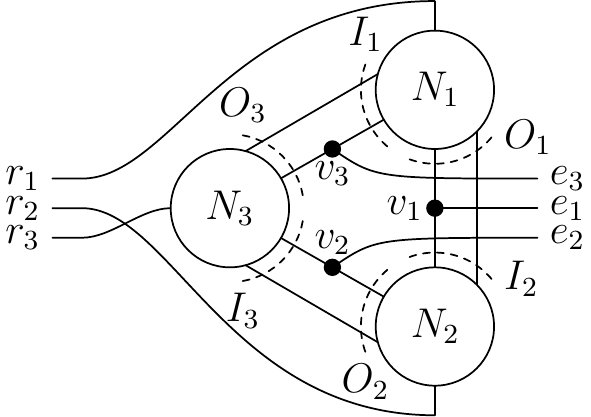}
	\caption{A $(3,3)$-pole $M_{24}$ contained in class 36-A snarks}
	\label{fig:36-A}
\end{figure}

Let $\phi$ be a colouring of
$M_{24}$ and assume that $\phi(I) = (a,b,c) \in \mathbb{K}^3$. It is not difficult to see that for the
connectors of the three negators we have either $\phisum(I_1) = \phi(I_2) =
\phisum(I_3) = 0$ or $\phisum(O_1) = \phisum(O_2) =
\phisum(O_3) = 0$; the proof is similar to that for
Class 34-D. If the former case occurs, then $\phisum(O_1) =
\phi(r_1) = a$ and the parity lemma implies that $\phi(e_3) =
\phisum(O_1)  + \phisum(I_2) = a$. In a similar manner we get $\phi(e_1)=b$ and  $\phi(e_2)=c$, 
whence  $\phi(O) = (b,c,a)$. If the latter case occurs, we similarly get that $\phi(O) = (c,a,b)$. 
If follows that each element of $\col(M_{24})$ is a $6$-tuple $(a, b, c, a', b', c')$ of elements of $\mathbb{K}$ where $(a', b', c')$ arises from $(a, b, c)$ by cyclically permuting its entries. (Obviously, certain $6$-tuples of this form may be absent in $\col(M_{24})$ when some of the negators are imperfect.)
All of such $6$-tuples are also contained in $\col(Y_2)$, except $(a, a, a, a, a, a)$. However, $(a, a, a, a, a, a) \notin \col(Y_3)$. Since $Y_2$ and $Y_3$ are colour-disjoint, so are $M_{24}$ and $Y_3$. Therefore $M_{24}*Y_3$ is a snark. With the exception of the Isaacs flower snarks this is the first family of snarks known to us where ordered connectors emerge in explaining uncolourability.

There are six snarks of order $36$ belonging to Class 36-A. 

\subsubsection*{Class 36-B}

Take five negators $N_i(I_i,O_i,r_i)$, where $i \in \{1,2,3,4,5\}$, arrange them in a cyclic manner,
and for $1\le i\le 4$ join $O_i$ to $I_{i+1}$. Let $I_1=\{i_1,i_2\}$ and $O_5=\{o_1,o_2\}$. Join $i_1$ to $o_1$, subdivide the resulting edge with a new vertex $v$, and attach to $v$ a dangling edge with semiedge $r_0$. Finally, join $r_0$ to $r_3$, $r_1$ to $r_5$, $i_2$ to $r_4$, and $o_2$ to $r_2$ to obtain a cubic graph denoted by $G$, see Figure~\ref{fig:36-2}.

\begin{figure}[h]
	\centering
	\includegraphics[]{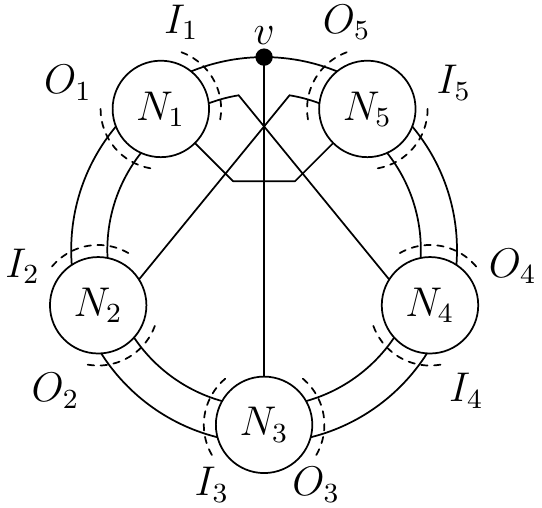}
	\caption{The structure of Class 36-B snarks}
	\label{fig:36-2}
\end{figure}

We show that $G$ is a snark. Suppose to the contrary that $G$ has
a $3$-edge-colouring~$\phi$. Since $N_3$ is a negator, 
the flow through one its connectors is nonzero. If $\phisum(O_3)\ne 0$, then there exists an element $a\in\mathbb{K}$ such that $\phisum(O_3)=a$. It follows that
$\phisum(O_4) = 0$ and $\phi(r_4) = a$. Moreover, $\phisum(I_2)\ne 0$, so $\phisum(I_1)=0$. Now, the edge leading from $N_4$ to $N_1$
through $I_1$ has colour $\phi(r_4)=a$, which implies that the edge 
connecting $N_1$ to $v$ through $I_1$ has the same colour. However, the edge
leading from $N_3$ to $v$ has colour $\phi(r_3)=a$, too, and this is a contradiction. If $\phisum(I_3)\ne 0$, a contradiction is derived similarly.

This family of snarks, built from five negators, can easily be generalised to a similar family built from a larger number negators, namely $4k+1$, where $k\ge 1$. In this case 
we connect the vertex $v$ to $N_1$ through $I_1$, to $N_{4k+1}$ through $O_{4k+1}$, and to $N_{2k+1}$ by using $r_{2k+1}$. We further join $i_2\in I_1$ with $r_{2k+2}$ and $o_2\in O_{2k+1}$ with $r_{2k}$. The proof that the resulting graph $G$ is a snark is similar to the one above. Clearly, $G$ can be made cyclically $5$-edge-connected by identifying pairs of residual semiedges appropriately.

Among the snarks of order $36$ we have identified two snarks
of this type.

\section{Strictly critical snarks}
\label{sec:notbicritical}

We conclude our investigation of small cyclically $5$-connected
critical snarks by turning our attention to those that are
strictly critical. In Section~\ref{sec:reducibility} we have
explained that strictly critical snarks are of special
interest, partially due to the fact that some of them can be
derived from noncritical snarks.

By Theorem~6.1 in \cite{Chladny-Skoviera-Factorisations}, a
strictly critical snark of order $n$ exists if and only if
$n\ge 32$. Among the snarks of order at most $36$, there are
only $84$ cyclically $5$-connected strictly critical snarks,
all having $36$ vertices. Of those, $77$ arise by an
NT-expansion from a non-critical snark of order $20$. The
structure of the remaining seven snarks is very similar: they
all arise from the Petersen graph by a substitution of a suitable 
proper $(2,2,2)$-pole $M$ for the $(2,2,2)$-pole $V_4$. The multipole $M$ is constructed as follows. Let $T_1$,
$T_2$, and $T_3$ be three perfect proper $(2,3)$-poles. Add
three new vertices and produce $M$ by connecting them to $T_1$,
$T_2$, and $T_3$ in the manner indicated in
Figure~\ref{fig:ttt-sc}. For the reasons similar to those valid
for the TTT $(2,2,2)$-pole (Section~\ref{sec:ttt}), the
multipole $M$ is colour-equivalent to the $(2,2,2)$-pole $V_4$
consisting of one vertex and its three neighbours. However, if
we remove from $M$ any two of the three new vertices, the
colouring set will not change. Consequently, no snark
containing $M$ is bicritical. The $(2,2,2)$-pole $M$ contained in all of the seven remaining snarks consist of three triads.

\begin{figure}
	\centering
	\includegraphics[]{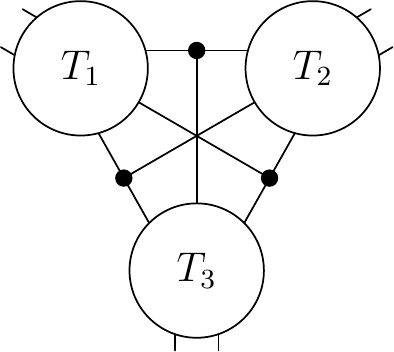}
	\caption{A $(2,2,2)$-pole contained in seven cyclically $5$-connected
	strictly critical  snarks of order $36$}
	\label{fig:ttt-sc}
\end{figure}

\section{An infinite family of bicritical snarks}
\label{sec:newfamily}

In the previous sections we have presented a number of new
constructions of snarks which mimic the structure of small
snarks. We now deal with the problem of their bicriticality.

As mentioned in the discussion of the results for order $36$,
the base components used in constructions of bicritical snarks
need not come from bicritical snarks. There are ten bicritical snarks of order $36$ that arise from noncritical snarks of order $24$ by a substitution of $P_{NN}$ for $P_2$. Obviously, it is necessary to
eliminate all pairs of removable vertices from the used
construction blocks, but in general we do not understand under
what conditions it would be sufficient.

On the other hand, taking all construction blocks from
bicritical snarks --- a slightly stronger requirement than just
the absence of removable pairs of vertices --- does not ensure
the resulting snark to be bicritical. For instance, a proper $(2,3)$-pole constructed even from a bicritical snark can be uncolourable (see results for order $26$). Such a proper $(2,3)$-pole cannot be used in any substitution which should yield a bicritical snark. In general, we do not know much about the circumstances under which proper $(2,3)$-poles are colourable or perfect. This is a significant difference from negators whose colouring properties are characterised by Theorem \ref{thm:negatos}.

The purpose of this section is to illustrate that imposing
certain additional requirements on the construction blocks can
assure bicriticality of the resulting snark in a fairly general
setting. The described requirements are not overly restrictive
and it is even possible that most construction blocks taken
from bicritical snarks (of any given order) satisfy them.

For our demonstration, we have chosen snarks constructed by an
NN-substitution (see Section~\ref{subs:NN}). This is perhaps  the simplest  class of 
infinite classes which we have described, nevertheless, a similar approach works for the
rest of them as well. We view the snarks arising by an $NN$-substitution as consisting of three
negators $N_i(I_i, O_i, r_i) = \Neg(G_i; u_i, v_i)$ for $i \in
\{1,2,3\}$ arranged along a circle with an additional vertex
attached to the residual semiedges (see Figure~\ref{fig:nnn2}).
We denote the resulting graph by $\NNN(N_1, N_2, N_3)$.

\begin{figure}[h!]
	\centering
	\includegraphics{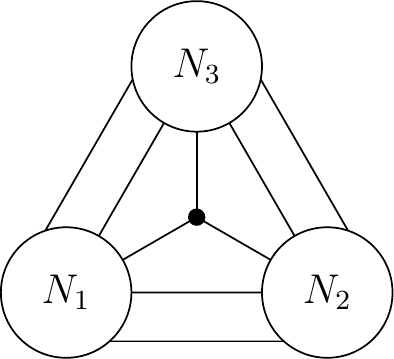}
	\caption{A schematic drawing of a snark $\NNN(N_1, N_2, N_3)$.}
	\label{fig:nnn2}
\end{figure}

As discussed in the beginning of this section, our restrictions
have to be imposed on the construction blocks, not on the snarks
they originate from. We will call a negator $N = \Neg(G; u, v)$
\emph{bicritical} if the multipole $G - \{x, y\}$ is colourable
for every two distinct vertices $x,\, y \in V(N)$. The
following proposition shows that this property is necessary.

\begin{proposition}\label{prop:bicrit-negators}
If $N_1$, $N_2$, and $N_3$ are negators such that $G=\NNN(N_1,N_2,N_3)$ is
a bicritical snark, then all of them  bicritical.
\end{proposition}

\begin{proof}
Consider a negator $N_i = N(G_i; u_i, v_i)$
and choose an arbitrary pair $\{x,y\}$ of its vertices. Since $G$ is a
bicritical snark, the multipole $G-\{x, y\}$ is colourable. By replacing
the $(2,2,1)$-pole $\NN(N_2, N_3)$ with the colour-equivalent
$(2,2,1)$-pole $P_2$,
we get the multipole $G_i - \{x,y\}$ which is
colourable as well. It follows that each negator $N_i$ is bicritical, as
claimed.
\end{proof}

We do not know whether the property stated in
Proposition~\ref{prop:bicrit-negators} --- or a stronger
property that all three negators are taken from bicritical
snarks --- is sufficient. We have constructed all the negators
from bicritical cyclically $5$-connected snarks on up to $30$
vertices. Using them we have created all possible snarks of
class NNN that contain at most two different negators; there
are approximately $600,\!000$ such snarks. With the help of a
computer we have verified that all of them are bicritical.

In order to specify a sufficient condition, we introduce the
following rather technical property of negators. We will
henceforth assume that the negators in question are constructed from snarks of girth at
least $5$, in order to avoid ambiguity of notation and certain marginal
cases.

\begin{definition}\label{def:goodnegator}
Let $G$ be a snark of girth at least $5$.
A negator $N(\{i_1,i_2\}, \{o_1, o_2\}, r) = \Neg(G; u, v)$ is called
\emph{feasible} if
it is bicritical and possesses the following properties:
\begin{itemize}
\item[(i)] For every pair of vertices $\{x,y\}$ where $x \in
    \{u, v\}$ and $y \in V(N)$ and for any two dangling
    edges $e$ and $f$ of the $6$-pole $G -\{x,y\}$ formerly
    incident with $x$ there exists a colouring $\varphi$ of
    the $6$-pole $G-\{x,y\}$ such that
    $\varphi(e)=\varphi(f)$.

\item[(ii)] For any vertex $y \in V(N)$ there exist
    colourings $\phi$ and $\psi$ of the $8$-pole $N-y$
    such that $\phi(i_1, i_2, o_1, o_2, r) =
    (a,a,b,b,a)$ and $\psi(i_1,  i_2, o_1, o_2, r) =
    (a,a,b,b,b)$ for some $a, b \in \K$ with $a \ne b$.
	\end{itemize}
\end{definition}

Here we regard each multipole $M = G - \{x, y\}$ always as 
a $6$-pole --- if the vertices $x$ and $y$ are adjacent,
we keep the edge $xy$ as an isolated edge in $M$. If $G$ is
bicritical, Property~(i) is easily satisfied.
Indeed, the $6$-pole $M$ is colourable because $G$ is
bicritical; by the parity lemma, the two dangling edges formerly
incident with $x$ have the same colour which can be assigned to
the isolated edge in $M$ so that all three semiedges formerly
incident with $x$ have the same colour for a suitable colouring of
$M$.

If $G$ is bicritical, $M$ is also colourable for each pair of
non-adjacent vertices $x$ and $y$ of $G$. From the parity lemma we
deduce that in every $3$-edge-colouring of $M$ two of the
semiedges formerly incident with $x$ must have the same colour.
Property~(i) requires more: the two semiedges having the
same colour can be arbitrarily prescribed. We have tested all
cyclically $5$-connected bicritical snarks of order up to $36$
and only six of them
support a negator violating Property (i). We describe them in
Section~\ref{sec:s36}.

From the definition, it is not clear whether a feasible negator has to be perfect or not. Amongst the negators constructed from nontrivial snarks up to order $28$, there is no example of an imperfect feasible negator.

Property~(i) of feasible
negators is related to a similar
technical property introduced by Chladný and Škoviera
\cite{Chladny-Skoviera-Factorisations} in their
study of criticality and bicriticality of dot products of
snarks.

\begin{definition}[\cite{Chladny-Skoviera-Factorisations}]
A pair $\{e,f\}$ of edges of a snark $G$ is called
\textit{essential} if it is non-removable and, moreover, if for
every $2$-valent vertex $v$ of $G-\{e,f\}$, the graph obtained
from $G-\{e,f\}$ by suppressing $v$ is colourable.
\end{definition}

\begin{lemma}
\label{lemma:essential-feasible}
Let $N = \Neg(G; u, v)$ be a snark, $x \in \{u, v\}$, and $y \in V(N)$. Assume that for every edge $e$ incident with $x$ in $G$ there is an edge $f$ incident with $y$ in $G$ such that $\{e, f\}$ is an essential pair of edges in $G$. Then the negator $N$ satisfies Property~{\rm (i)} of Definition~\ref{def:goodnegator}. 
\end{lemma}

\begin{proof}

Let $e_1$, $e_2$, and $e_3$ denote the edges incident with $x$ listed in  an arbitrary order. 
From our assumption it follows that there is an edge $f$ incident with $y$ such that $\{e_1, f\}$ is an essential pair of edges in $G$. Therefore, the multipole $G - \{e_1, f\}$ with $x$ suppressed has a
colouring $\varphi$. If we cut the edge resulting from the
suppression of $x$ into two dangling edges corresponding to
$e_2$ and $e_3$ and remove $y$, we get the $(3,3)$-pole $G -
\{x, y\}$ with a colouring in which the dangling edges
corresponding to $e_2$ and $e_3$ have the same colour, where $e_2$ and $e_3$ can be chosen arbitrarily. Thus the
negator $N$ satisfies Property~(i) of Definition~\ref{def:goodnegator}.
\end{proof}

We have tested negators for Property~(ii). Although
there are many negators violating (ii), more than $90\%$ of all
negators created from bicritical cyclically $5$-connected snarks with at most $34$
vertices are feasible. For instance, for every such snark $G$ of order $34$, there are, if we
ignore possible isomorphisms, $102$ possible negators which can be constructed
from $G$; the number of the feasible ones among them ranges from $74$ to $102$.

\begin{theorem}
If $N_1$, $N_2$, and $N_3$ are any three feasible perfect
negators, then $G=\NNN(N_1, N_2, N_3)$ is a~bicritical snark.
\end{theorem}

\begin{proof}
For $j \in \{1,2,3\}$ let  $N_j= \Neg(G_j; u_j, v_j)$.
Let $I_j= \{i_j, i_j'\}$ and $O_i = \{o_j, o_j'\}$ be the connectors of $N_j$, let
$r_j$ be its residual edge, and let $w_j$ be the the common neighbour of $u_j$ and $v_j$ in $G_j$.
Choose any two distinct vertices  $x$ and $y$ of $G$. We
show that the multipole $G - \{x, y\}$ is colourable, implying
that $G$ is bicritical. 	
	
\medskip
	
\noindent\underline{Case 1.} If both vertices $x$ and $y$
belong to the same negator $N_j$, we can replace the other two
negators with the colour-equivalent path $P_2$ (path of length
two), completing the negator $N_j$ into a snark $G_j$. Since
$N_j$ is bicritical, $G_j - \{x, y\}$ is colourable, hence so
is $G - \{x, y\}$.
	
\medskip
	
\noindent\underline{Case 2.} Assume that $x$ and $y$ belong to
different negators, say, $x \in V(N_1)$ and $y \in V(N_2)$.
Remove the vertices $v_1$ and $x$ from the snark $G_1$ and
denote the semiedges formerly incident with $v_1$ by $e_1$,
$e_2$, and $e_3$ in such a way that  $e_3$ is incident with $w_1$. According
to Property~(i) of the feasible negator $N_1$, there exists a
colouring $\varphi_1$ of $G_1 - \{v_1, x\}$ such that
$\varphi_1(e_1) = \varphi_1(e_2)$. Let $a = \varphi_3(e_3)$ and
$b = \varphi_1(u_1w_1)$; obviously $a\neq b$ (see
Figure~\ref{fig:nnn-col1}). We can simply restrict the
colouring $\varphi_1$ to a colouring of the multipole $N_1 - x$
for which $\varphi_{1*}(I_1) = \varphi_1(u_1w_1) = b$,
$\varphi_1(r_1) = a + b$ and $\varphi_{1*}(O_1) = 0$.
	
	\begin{figure}[!h]
		\centering
		\begin{minipage}[c]{0.45\textwidth}
			\centering
			\includegraphics[]{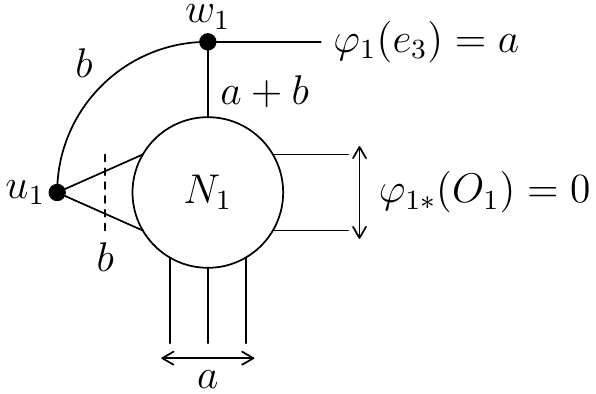}
			\captionof{figure}{The colouring $\varphi_1$ of $N_1$.}
			\label{fig:nnn-col1}
		\end{minipage}
		\begin{minipage}[c]{0.45\textwidth}
			\centering
			\includegraphics[]{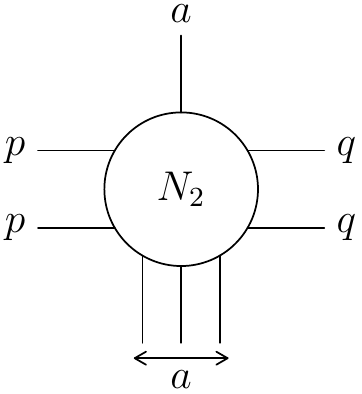}
			\captionof{figure}{The colouring $\varphi_2$ of $N_2$.}
			\label{fig:nnn-col2}
		\end{minipage}
	\end{figure}
	
Let $p = \varphi_1(o_1)$. The negator $N_2$ is feasible, so
according to Property (ii) there exists a colouring $\varphi_2$ of $N_2- y$ such that
$\varphi_2(r_2) = a$, $\varphi_2(i_2) = \varphi_2(i_2') = p$, and $\varphi(o_2) = \varphi(o_2') = q$. If $p = a$, then $q = b$; if $p \ne a$ then $q = a$. The colourings $\varphi_1$ and $\varphi_2$ are compatible and can be glued together to form a colouring $\varphi$ of the
multipole $M = \NN(N_1, N_2) - \{x,y\}$ depicted in
Figure~\ref{fig:nnn-col}. Since $N_3$ is perfect, the colouring
$\varphi$ of $M$ can be extended to the entire $6$-pole 
$G - \{x,y\}$. 	
	
	\begin{figure}
		\centering
		\includegraphics[]{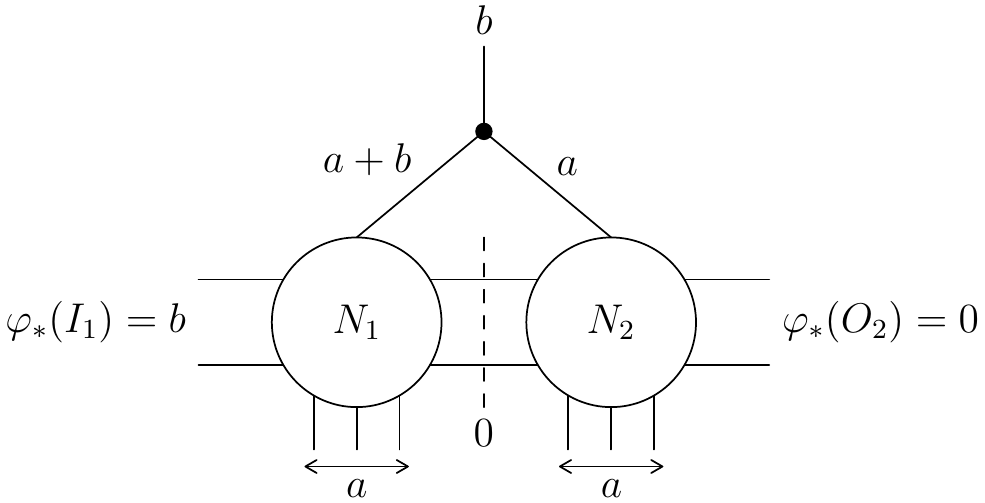}
		\caption{Colouring of the $11$-pole $\NN(N_1, N_2) - \{x, y\}$}
		\label{fig:nnn-col}
	\end{figure} 	
	
\medskip 	\noindent\underline{Case 3.} If one of the
vertices $x$ and $y$, say $y$,  does not belong to any of the negators,
it must be the vertex attached to the residual semiedges of
$N_1$, $N_2$, and $N_3$. Let $x \in V(N_1)$. Property~(ii) applied to 
the feasible negator $N_1$ guarantees there exists a colouring
$\varphi_1$ of $N_1 - x$ such that $\varphi_1(i_1) =
\varphi_1(i_1') = a$ and $\varphi_1(r_1) = \varphi(o_1) =
\varphi_1(o_1') = b \ne a$. Since $N_2$ is perfect, it admits a
colouring $\varphi_2$ such that $\varphi_{2*}(I_2) = 0$ and
$\varphi_{2*}(O_2) \neq 0$. Finally, $N_3$ is also perfect, and
hence it admits a colouring $\varphi_3$ that is compatible with
both $\varphi_1$ and $\varphi_2$. The partial colourings
$\varphi_1$, $\varphi_2$, and $\varphi_3$ can be combined to
colour the entire $6$-pole $G - \{x, y\}$. 
This completes the proof of the theorem.
\end{proof}

In order to create an infinite class of bicritical snarks, we
need an infinite family of feasible negators. As one could
expect, all the negators constructed from Isaacs snarks are
feasible. Our proof is based on the fact, established in  \cite{Chladny-Skoviera-Factorisations}, that every pair of
non-adjacent edges in any Isaacs snark except $J_3$ is
essential.

\begin{proposition}
Every negator constructed from the snark $J_k$ is feasible.
\end{proposition}

\begin{proof}
Since every pair of edges of $J_k$, with $k\ge 5$, is essential, it follows from Lemma~\ref{lemma:essential-feasible} that each negator of $J_k$ has 
Property~(i) of a feasible negator.
The fact that these negators possess also Property~(ii) will be
established by induction on $k$. 

For the induction basis we have checked that all 
negators derived from the snarks $J_5$,
$J_7$, and $J_9$ satisfy Property~(ii); we did it with the help of a 
computer.

Now, consider the Isaacs
snark $J_k$ for an odd $k \ge 11$. Remove an arbitrary path
$uwv$ from $J_k$ to produce a negator $N = \Neg(J_k; u, v)$ whose dangling edges are denoted by $i_1$, $i_2$, $o_1$, $o_2$ and $r$ in the usual
way. Remove from $N$ an arbitrary vertex $x$ and denote the resulting 
$8$-pole by $M$.
The path $uwv$ intersects at most three consecutive copies
of the Isaacs $(3,3)$-pole $Y$ and the removal of the vertex
$x$ corrupts at most one other copy of $Y$. Consequently, there
are at least four consecutive copies of $Y$ in $M$ that remain intact; let
$Y_4$ denote the $(3,3)$-pole which they induce. We replace them
with a $(3,3)$-pole $Y_2$ consisting of two copies of $Y$ and
denote the resulting multipole $M'$. Clearly, $M'$ is
isomorphic to the multipole obtained from $J_{k-2}$ by removal
of a certain path of length two and a certain additional
vertex. By the induction hypothesis, there exists a colouring
of $M'$ in which the dangling edges corresponding to $i_1$,
$i_2$, $o_1$, $o_2$, $r$ have colours exactly as desired for
either $\varphi$ or $\psi$ from Property~(ii). Since
the multipoles $Y_4$ and $Y_2$ are colour-equivalent
\cite{Nedela}, the desired colours can also be assigned to the
semiedges $i_1$, $i_2$, $o_1$, $o_2$, $r$ of $M$. Hence, any
negator constructed from the Isaacs snark $J_k$ satisfies Property~(ii).
Consequently, every negator constructed from the snark $J_k$ is feasible, as claimed.
\end{proof}

Using feasible negators from the Isaacs snarks $J_k$, with $k \ge
5$, we can construct an infinite class of bicritical snarks.
All such snarks are cyclically $5$-connected. Since $5$-cycles can only occur in negators constructed from $J_5$, avoiding such negators will lead to bicritical snarks of girth $6$.

\section{Non-removable edges that are not essential}
\label{sec:s36}

As anticipated in the previous section, we wish to
investigate snarks containing negators that violate
Property~(i) of Definition~\ref{def:goodnegator}. Property~(i) of a
feasible negator is related to the concept of an essential pair of
edges, which in turn plays a crucial role in factorisation of 
a critical snark into a dot product of two smaller snarks (see
Theorem~\ref{thm:A}).

\begin{figure}[h!]
	\centering
	\includegraphics[]{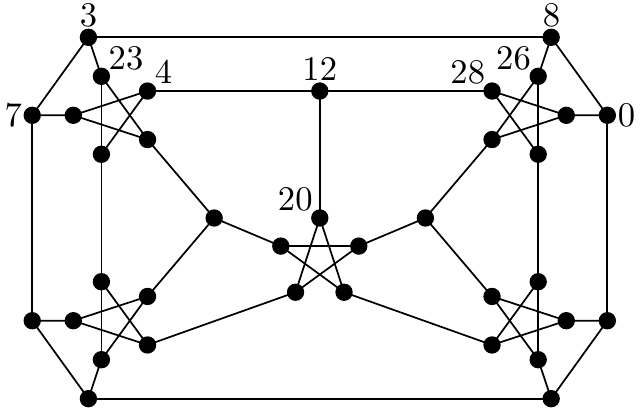}
	\caption{The snark $G_{36}$.}
	\label{fig:36-no2}
\end{figure}

Among the $5$-simple snarks of order up 36 there are exactly six snarks supporting a negator violating Property~(i), all of them belonging to the class NNN. They consist
of two Petersen negators and one negator constructed from a reducible snark
of order $24$ (see Figure \ref{fig:24}). One of these snarks,
denoted by $G_{36}$, is illustrated in Figure \ref{fig:36-no2}.
If we remove the pair of vertices $12$ and $3$ (or $12$ and $8$), we
get a $6$-pole $M$ such that for each colouring of $M$, the
dangling edges incident with the vertices $20$ and $4$ (or $20$ and
$28$) have different colours; this property has been verified
by exhaustive computer search. If we construct a negator
from the snark $G_{36}$ by removing a path of length $2$
starting from the vertex $12$, the result violates Property~(i).

The snark $G_{36}$ has another interesting property. If we take
the $6$-pole $G_{36} - \{3, 12\}$ and perform the junction
of the semiedges $(4)$ and $(20)$, we get a $4$-pole which is
uncolourable (because the two joined semiedges have different
colours in any possible $3$-edge-colouring). Furthermore, we
can add one vertex incident with semiedges $(8)$ and $(23)$;
the resulting uncolourable multipole is isomorphic to
$G_{36} - \{12\text{-}28, 3\text{-}7\}$ with the vertex $12$ suppresed. This implies
that the pair of edges $\{12\text{-}28, 3\text{-}7\}$ is not
essential in $G_{36}$. On the other hand, with the help of a
computer we have found a colouring which proves that this pair
of edges is non-removable. This solves Problem 5.7 proposed by
Chladný and Škoviera in \cite{Chladny-Skoviera-Factorisations}
by showing that there exists a pair of non-removable edges in
an irreducible snark which is not essential. The same holds for
the pairs of edges $\{12\text{-}28, 3\text{-}23\}$,
$\{12\text{-}4, 8\text{-}0\}$ and $\{12\text{-}4,
8\text{-}2\}$.

\section{Beyond order 36}

We conclude our paper by analysing the currently known $5$-simple snarks
of order $38$. We also explain how our results can be used to generate some of such snarks of higher orders.  

At present, there are $19,775,768$ known nontrivial snarks of
order 38 (see \cite{BCGM}, section Snarks). 
Of them, $56$ are $5$-simple snarks, all being 
bicritical. In the latter set we have identified the $5$-pole
$P_{NN}$ in $39$ snarks, the $5$-pole $P_{NT}$ in $22$ snarks,
and the $5$-pole $P_{TT}$ in $7$ snarks, while there are $10$
snarks containing both $P_{NN}$ and $P_{NT}$ and $6$ snarks
containing both $P_{NN}$ and $P_{TT}$. 
In three snarks we have found the panchromatic $(2,2,2;1)$-pole $M_{3NT} =
\NNNT(\dyad,\dyad, \dyad, \triad)$ (see Section~\ref{sec:panchromatic}). All three snarks arise from the 
Petersen graph by substitution of $M_{3NT}$ for a path of length two and one edge.
The only remaining snark gives rise to an infinite family, which we are just about to describe.

\subsection*{Class 38-A}

Take four negators $N_i(I_i,O_i;r_i)$, for $i \in
\{1,2,3,4\}$, and one proper $(2,3)$-pole $T(B,C)$, connect them
as shown in Figure~\ref{fig:38-A}, and denote the resulting
graph by $G$. 

\begin{figure}[h]
	\centering
	\includegraphics[]{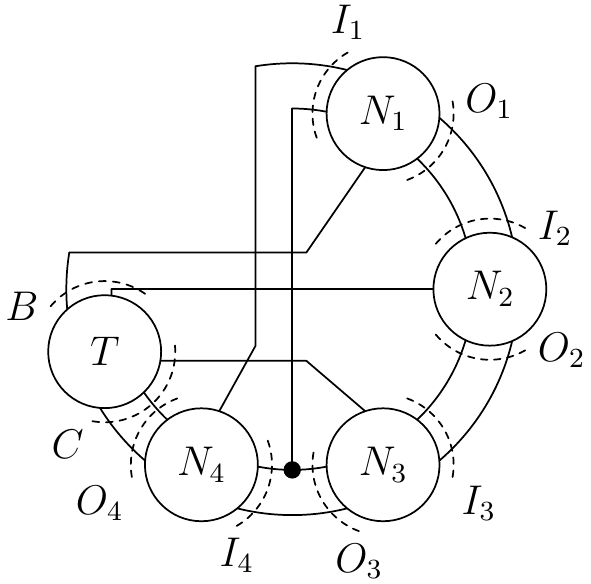}
	\caption{The structure of Class 38-A snarks}
	\label{fig:38-A}
\end{figure}

We show that $G$ is a snark.
The connector $B$ is proper, so the edges $r_1$ and $r_2$
contained in $B$ have different colours, say $\phi(r_1)
= a$ and $\phi(r_2) = b$,  where $a \ne b$. It follows that
$\phisum(O_1) = \phisum(I_2) = 0$, for otherwise we would have $a=b$. Moreover, $\phisum(I_1) =a$ and $\phisum(O_2) = \phisum(I_3) = b$. Since $N_3$ is a negator,
we infer that $\phi(r_3) = b$ and
$\phisum(O_3) = 0$. Knowing the colour of three semiedges of
the proper $(2,3)$-pole $T$, 
we can use the Kirchhoff law to determine the flow through the remaining two semiedges, which coincides with the flow through $O_4$. We have 
$\phisum(O_4) = \phi(r_1) + \phi(r_2)+ \phi(r_3) = a + b + b = a \ne 0$,  so
$\phisum(I_4) = 0$. 
Therefore $\phi(r_4)=a=\phisum(I_1)$, which 
forces the flow value on the edge of $I_1$ different from $r_4$ to be zero. This contradiction proves that $G$ is a snark.

The one remaining $5$-simple snark of order $38$ consists of four dyads and one triad. One can clearly see that permuting the semiedges in the connector $C$ of the triad leads to other $5$-simple snarks of order $38$. Moreover, permutations of the semiedges in the connectors of dyads may also give rise to additional nonisomporphic $5$-simple snarks of order~$38$.

\subsection*{Class 42-A}

This family illustrates the fact mentioned in Section~\ref{sec:methods} that, unlike quasitriads, both triple pentagons and tricells may occur in $5$-simple snarks. From among the six $5$-cycle clusters on up to 10 vertices shown in Figure~\ref{fig:42-a} quasitriad is thus the only one that cannot occur in $5$-simple snarks.

Let $R$ be a $(2,2,2)$-pole obtained from a snark $G$ by severing three pairwise nonadjacent edges; if $G$ is the Petersen graph, then $R$ is either the triple pentagon or the tricell, depending on the choice of the three edges.
Let $N_1$ and $N_2$ be two negators, and $T_1$ and $T_2$ two proper $(2,3)$-poles. Construct a $10$-pole $M$ as depicted in Figure~\ref{fig:42-a}. The $10$-pole $M$ is uncolourable: for each $i \in \{1, 2\}$ the flow between $N_i$ and $R$ has to be zero, which is impossible for $R$. Thus, to obtain a snark it is sufficient to perform junctions of the $10$ semiedges of $M$. If we take all the multipoles from the Petersen graph, we obtain several $5$-simple snarks of order $42$. Somewhat surprisingly, it is also possible to perform the junctions in such a way that the outcome will be cyclically $5$-edge connected, but not critical.

\begin{figure}[h]
	\centering
	\includegraphics[]{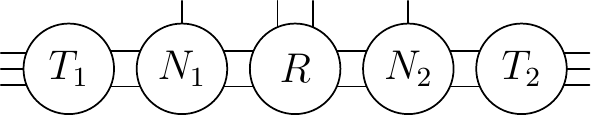}
	\caption{A new infinite class of snarks containing snarks with triplepentagons and tricells.}
	\label{fig:42-a}
\end{figure}

\subsection*{Constructions and analysis of larger critical snarks}

Infinite families described in our paper can be used to generate new cyclically $5$-connected critical snarks of orders greater than $36$.  We briefly sketch the ideas that can be used to construct reasonable amounts of such snarks.

To construct a member of an infinite family we choose 
snarks for the construction of the desired multipoles. Cyclically $5$-connected critical snarks are a sensible choice, nevertheless, we have seen examples where the multipoles were constructed from non-critical snarks or from snarks with smaller cyclic connectivity. These snarks might produce new members of the families, however, it requires more computational time.
One way or another, we still need to check if the resultant snarks are cyclically $5$-connected and critical, which is not guaranteed just by the membership in any of our families.
Another possibility to construct new families of snarks that can contain cyclically $5$-connected critical ones is to use ideas which occurred in our proofs of uncolourability.

Currently we are not able to construct any sort of a complete list of snarks of order 38 (or more). If a complete list of all the cyclically $5$-connected critical snarks of a given order $n \ge 38$  eventually becomes available, it is possible to employ our methods to analyse them. We expect that, at least for reasonably small $n$, a large number of these snarks will fall into one of already described families, and the remaining snarks give a rise to new infinite families of snarks. With increasing $n$, however, the approach based on the analysis of $5$-cycle clusters will become less efficient. For instance, we would not be able to identify important multipoles, such as negators or proper $(2,3)$-poles, arising from Isaacs flower snarks. It might be therefore useful to extend the analysis of the computer generated snarks by including the search for certain subgraphs of flower snarks, for instance the iterated Isaacs $(3,3)$-poles $Y_k$ (which can be regarded as clusters of $6$-cycles). In general, however, there is no known method for the analysis of cyclically $6$-connected snarks. In spite of the efforts of Karab\'a\v s et al. \cite{Karabas}, a decomposition theorem for cyclically $6$-connected snarks similar to decomposition theorems for lower connectivities, such as those proved in \cite{Cameron, Nedela}, is not known. Furthermore, small cyclically $6$-connected snarks different from the flower snarks seem to be very difficult to find: the smallest known example was constructed in 1996 by Kochol in \cite{Kochol_118} and has 118 vertices.

\bigskip
\noindent\textbf{Acknowledgements.} The authors acknowledge
partial support from the research grants APVV-19-0308, VEGA
1/0813/18 and VEGA 1/0876/16. 

\bibliographystyle{siam}

\newcommand{\noopsort}[1]{} \newcommand{\printfirst}[2]{#1}
\newcommand{\singleletter}[1]{#1} \newcommand{\switchargs}[2]{#2#1}

\end{document}